\documentclass[11pt]{article}
\usepackage[utf8]{inputenc}
\usepackage[camera]{dtrt}
\usepackage[left=1in, right=1in, bottom=1.15in, top=1.1in]{geometry}

\usepackage{xcolor}
\definecolor{Cerulean}{rgb}{0.0, 0.48, 0.65}
\definecolor{MidnightBlue}{rgb}{0.1, 0.1, 0.44}
\definecolor{NavyBlue}{rgb}{0.0, 0.0, 0.5}

\usepackage{amsmath,amssymb}
\usepackage{mathtools}
\usepackage{hyperref}
\hypersetup{%
  colorlinks=true,
  citecolor=Cerulean,
  linkcolor=Cerulean,
  urlcolor=Cerulean,
}
\usepackage{bbm}
\usepackage{cleveref}
\usepackage{enumitem}
\usepackage{makecell}
\usepackage{bm}

\usepackage[vlined,ruled,linesnumbered]{algorithm2e} 
\SetKwRepeat{Do}{do}{while}

\usepackage{verbatim} %
\usepackage{multirow}
\usepackage{amsthm} %
\newtheorem{theorem}{Theorem}

\newtheorem{lemma}{Lemma}

\newtheorem{definition}{Definition}
\newtheorem{claim}{Claim}

\newtheorem{example}{Example}

\usepackage{threeparttable}
\usepackage{graphicx}
\usepackage{subcaption}
\usepackage{tikz}
\usetikzlibrary{automata, positioning, arrows}
\tikzset{
box/.style ={
rectangle, %
rounded corners =1pt, %
minimum width =50pt, %
minimum height =20pt, %
inner sep=5pt, %
draw=black, %
align = center
}}

\newcommand{\calD}{\mathcal{D}}

\newcommand{\calX}{\mathcal{X}}
\newcommand{\calY}{\mathcal{Y}}

\colorlet{entry}{Cerulean}
\colorlet{string}{cyan}
\definecolor{mypurple}{HTML}{57068c}
\newcommand{\protoAMM}{\mathbf{Prot}_\AMMName}
\newcommand{\puser}{\textcolor{mypurple}{\mathcal{P}}}
\newcommand{\parbtrgr}{\textcolor{mypurple}{\mathcal{A}}}
\newcommand{\stringlitt}[1]{\textcolor{string}{\text{``#1''}}}
\newcommand{\msgorder}{\stringlitt{order}}
\newcommand{\msgarbitrage}{\stringlitt{arbitrage}}
\newcommand{\onrecv}{\textcolor{entry}{\bf On receive}\xspace}
\newcommand{\ontime}{\textcolor{entry}{\bf On time}\xspace}
\newcommand{\MCMsettle}{\textcolor{magenta}{\bf settle}\xspace}
\usepackage{boxedminipage}
\usepackage[ff,adversary,keys,primitives,operators,lambda,logic,notions]{cryptocode}
\renewcommand{\pccomment}[1]{\textcolor{gray}{\scriptsize // #1}}

\newcommand{\protocolsidebyside}[3]{
\begin{boxedminipage}[t]{\textwidth}
	\begin{center}
	\scriptsize{\textbf{#1}}
	\end{center}
	\vspace{-\baselineskip}
	\begin{pchstack}[center]
		\procedure[mode=text, linenumbering, codesize=\footnotesize]{}{
		    #2
		}
		\procedure[lnstart=21,mode=text, linenumbering, codesize=\footnotesize]{}{
			#3
		}
	\end{pchstack}
\end{boxedminipage}
}

\def\TX{\textsf{TX}}
\def\MEV{\textsf{MEV}}
\def\SB{\textsf{SB}}

\newcommand{\AMMName}{\textsf{RediSwap}\xspace}
\newcommand{\cowswap}{CoWSwap\xspace}
\newcommand{\uniswapx}{UniswapX\xspace}
\newcommand{\set}[1]{\left \{ #1 \right \} }

\title{\AMMName: MEV Redistribution Mechanism for CFMMs}

\author{
Mengqian Zhang
\\\textit{Yale University} \\
\texttt{mengqian.cs@gmail.com}
\and Sen Yang \\\textit{Yale University}\\ \texttt{sen.yang@yale.edu}
\and Fan Zhang \\\textit{Yale University}\\ \texttt{f.zhang@yale.edu}
}

\date{}

\begin{document}

\maketitle
\begin{abstract}

Automated Market Makers (AMMs) are essential to decentralized finance, offering continuous liquidity and enabling intermediary-free trading on blockchains. However, participants in AMMs are vulnerable to Maximal Extractable Value (MEV) exploitation. Users face threats such as front-running, back-running, and sandwich attacks, while liquidity providers (LPs) incur the loss-versus-rebalancing (LVR).

In this paper, we introduce \AMMName, a novel AMM designed to capture MEV at the application level and refund it fairly among users and liquidity providers. 
At its core, \AMMName features an MEV-redistribution mechanism that manages arbitrage opportunities within the AMM pool. 
We formalize the mechanism design problem and the desired game-theoretical properties. 
A central insight underpinning our mechanism is the interpretation of the maximal MEV value as the sum of LVR and individual user losses. We prove that our mechanism is incentive-compatible and Sybil-proof, and demonstrate that it is easy for arbitrageurs to participate. 

We empirically compared \AMMName with existing solutions by replaying historical AMM trades. Our results suggest that \AMMName can achieve better execution than \uniswapx in 89\% of trades and reduce LPs' loss to under 0.5\% of the original LVR in most cases.

\vspace{0.5cm}
\noindent\textbf{Keywords:} Decentralized Finance, MEV Redistribution, Mechanism Design
\end{abstract}

\section{Introduction}
Automated Market Makers (AMMs) have become the leading design for facilitating trades on decentralized exchanges (DEXs). A key feature of AMMs is their use of liquidity pools, which are composed of tokens contributed by liquidity providers (LPs) and accumulated from past trades. This structure enables users to trade directly with the pool, eliminating the need for order matching as seen in the traditional order book system used by centralized exchanges (CEXs). %

Despite various benefits, two types of participants in AMMs—users and liquidity providers—suffer from the phenomenon known as Maximal/Miner Extractable Value (MEV)~\cite{daian2020flash}. MEV refers to the profit that intermediaries, such as searchers, builders, and proposers, can extract during the block production. Because trades in DEXs are publicly visible in the mempool before they are confirmed, those monitoring the mempool (e.g., searchers) can profit by front-running, back-running, or sandwiching user trades~\cite{qin2022quantifying,torres2021frontrunner,zhou2021high}. As a result, user trades execute at a worse price. For liquidity providers, the primary risk comes from Loss-Versus-Rebalancing (LVR)~\cite{milionis2022automated}, which quantifies the cost they incur when arbitrageurs rebalance the AMM pools in response to price movements at external markets. The existence of LVR makes it challenging for liquidity providers to earn sustainable profits, even with swap fees.

One promising direction to mitigate users' loss is \textit{refunding} them partial MEV extracted from their trades. MEV-Share~\cite{flashbots2024mevshare} and MEV Blocker~\cite{MEVBlocker} are primary examples used in practice.
However, these refunding mechanisms are rather limited, mainly because they operate near the end of the MEV supply chain, after the trades have been exploited by various intermediaries, such as searchers/arbitrageurs and builders.
Moreover, some (e.g., application-level) information is lost as transactions pass through the supply chain. One consequence is unfair redistribution. For instance, in MEV Blocker, the majority of refunds for \cowswap~\cite{cow2022protocol} orders were erroneously delivered to solvers rather than users~\cite{cowprotocol2023mevblocker}. 

On the other hand, \uniswapx~\cite{uniswapx2023overview} and \cowswap are two notable solutions at the application level, an \textit{upstream} stage of the supply chain. In both systems, users submit intents specifying their desired outcomes, and a market of solvers compete to search for the best settlement. These schemes are expected to provide better execution because solvers can aggregate liquidity from various sources and access more resources (information, capital, and sophisticated execution strategies) than users.
However, empirical evidence (as detailed in~\cref{sec:empiricalevidence}) suggests that solvers may not exhibit the expected level of sophistication, emphasizing the need for simpler mechanism designs.

Several parallel works have been proposed to mitigate LVR and compensate liquidity providers. 
A widely studied approach is to auction off exclusive rights earlier before block generation in exchange for compensation to LPs~\cite{adams2024amm,josojo2024McAMM}. 
A major concern with these \textit{ex-ante} auctions is that arbitrageurs must bid based on their estimation of future profits, which makes it challenging to engage and discourages risk-averse players. Additionally, other LVR mitigation ideas have been explored, such as reducing the inter-block time and involving dynamic fees~\cite{fritsch2024measuring,milionis2023effect}. However, like all previous MEV mitigation solutions, these approaches focus solely on protecting the interests of one group (LPs in this case), while do not take the other group of participants into account.

The limitations of existing solutions motivate us to explore AMM designs that can capture MEV at the application level, redistribute MEV fairly among users and LPs, and ensure ease of participation for solvers (arbitrageurs). Particularly, we tackle the research question through the lens of mechanism design.

\subsection{This Work}
In this paper, we propose an MEV-redistribution Constant Function Market Maker (CFMM)~\cite{angeris2020improved} called \AMMName.
\AMMName addresses the above limitations: the refund mechanism takes both users and LP into consideration; the arbitrageurs only need to take simple actions to participate. Looking ahead, our empirical evaluation suggests that \AMMName achieves better execution results than existing systems while mitigating LVR.

\AMMName has two main components: a CFMM and an MEV-redistribution mechanism that manages arbitrage opportunities\footnote{We focus on non-atomic arbitrage, which capitalizes on price discrepancies between on-chain DEXs and exchanges in another venue (i.e., CEXs like Binance or DEXs on other blockchains). Non-atomic arbitrage is one of the dominant forms of MEV. Empirical studies indicate that over a quarter of the trading volume on Ethereum's biggest five DEXs is likely attributed to non-atomic arbitrage~\cite{heimbach2024non}, and since the Ethereum Merge~\cite{ethereum2024merge}, the total profits from CEX-DEX arbitrage have reached \$131.77M~\cite{sorella2024mevdashboard}. } within the CFMM pool. We focus on the CFMM with a risky asset and a num\'{e}raire. The CFMM follows the standard design, except that users send trades to the MEV-redistribution mechanism, which uses the CFMM to execute them; the CFMM does not accept trades directly from users.

The overall workflow of \AMMName is as follows: users privately send transactions to the MEV-redistribution mechanism, e.g., via an RPC channel. 
Arbitrageurs interested in the potential arbitrage opportunity provide the MEV-redistribution mechanism with their beliefs regarding the external market price of the risky asset. The MEV-redistribution mechanism is essentially an \textit{ex-post} auction to sell arbitrage opportunities to arbitrageurs, specifying three key rules: (1) a bundle generation rule, which constructs a list of transactions (i.e., a bundle) with optimal arbitrage profit based on the inputs from users and arbitrageurs; (2) a payment rule, which decides arbitrageurs' payments, capturing a portion of the MEV within the bundle; and (3) a refund rule, which rebates the captured MEV among users and liquidity providers.

While \AMMName relies on arbitrageurs (as with \uniswapx and \cowswap), arbitrageurs only need to take simple actions, i.e., to provide their beliefs of the external prices and capital. 
That is, \AMMName internalizes the complexity of computing optimal arbitrage into the mechanism.
Moreover, \AMMName has full control over transaction ordering, which enables transparent MEV capturing and can enforce fair redistribution. Unlike MEV-Share or MEV Blocker, \AMMName does not expose transaction information to arbitrageurs.

\subsection{Our Contributions}
We present the first truthful and Sybil-proof MEV-redistribution mechanism for CFMMs.

In~\cref{sec:model}, we formalize the problem of designing an MEV-redistribution mechanism and define the desired properties. Intuitively, we aim for a mechanism that is (for arbitrageurs) \textit{incentive-compatible} (in the sense that truthfully reporting beliefs regarding the price of the risk asset is a dominant strategy) and \textit{Sybil-proof} (in the sense that creating Sybil transactions will not decrease any user's utility). Such a mechanism would make it simple for arbitrageurs to participate and not require much sophistication.

Capturing MEV and refunding it to every player requires us to understand not only how to maximize MEV, but also each player's contribution to the MEV. Thus in \cref{sec:optimalMEV}, we revisit a simpler problem, namely, how to construct an optimal non-atomic arbitrage strategy in the scenario with public order flows. We answer this question from a new perspective using potential functions. Based on the potential function characterization,  we argue that the optimal MEV value can be interpreted as the sum of LVR and each user's loss so that \emph{we can quantify the loss of each player}. 

This insight naturally leads to a strawman MEV-redistribution mechanism that sells all MEV opportunities to a single arbitrageur and refunds the winner's payment to users and liquidity providers proportional to each player's loss (\cref{sec:strawman}). Since the MEV-redistribution mechanism is designed to operate in an open and decentralized environment, arbitrageurs can pose as users and mount \textit{Sybil attacks}. The strawman solution would allow arbitrageurs to submit ``fake'' transactions to steal users' refunds. As a result, such a mechanism is shown to be truthful but, unfortunately, fails to be Sybil-proof.

The key idea in \AMMName (\cref{sec:ourmechanism}) is to sell the MEV opportunities from each pending transaction and the rebalancing arbitrage separately. The challenge is that the execution of a trade in CFMM is sensitive to the ordering, and so is the MEV value it creates. \AMMName solves this problem by ensuring independence among transactions (and corresponding refunds).
Theoretically, we prove that Sybil attacks will not reduce any user's utility, whether being truthful or not (\Cref{thm:sybilproof}). Moreover, if an arbitrageur in \AMMName has the required resource (i.e., tokens) to mount Sybil attacks, there exists a Sybil strategy such that using this strategy and truthfully reporting is a Nash equilibrium (\Cref{thm:NE}); otherwise, truthfully reporting is a dominant strategy (\Cref{theorem:truthful}).

In~\cref{sec:evaluation}, we evaluate our mechanism by replaying historical AMM trades from September 1st, 2023, to August 31st, 2024. To simulate arbitrageurs' beliefs in external prices, we use prices in Binance (a centralized exchange) to build price distribution.
By comparing the execution prices of the same order on \uniswapx or \cowswap with \AMMName, we show that our mechanism can achieve better execution prices than \uniswapx in 89\% of cases, and comparable execution prices to \cowswap. Our evaluation also shows that \AMMName effectively reduces LVR for the two Uniswap v2 liquidity pools (WETH-USDC and WETH-USDT) we experiment with, and the efficiency strengthens as more arbitrageurs participate.
For example, WETH-USDC liquidity providers incur less than 0.5\% of the LVR they would face without \AMMName in most cases.

\section{Related Work}
\parhead{MEV mitigation through AMM design} 
Several works focus on composing networks of AMMs to achieve better settlement and minimize arbitrage and slippage~\cite{angeris2022optimal,engel2021composing,zhou2021a2mm}. Instead of executing transactions sequentially, some designs~\cite{cow2022protocol,ramseyer2023speedex} process transactions in batches at a uniform clearing price, eliminating the risk of cyclic arbitrage and sandwich attacks within a block, though not all forms of manipulation are fully addressed~\cite{zhang2024computation}. Some studies have also explored the idea of integrating batch trading with CFMMs~\cite{canidio2023batching,ramseyer2022augmenting}. Orthogonally, Ferreira and Parkes~\cite{xavier2023credible} initiate the study of verifiable sequencing rules and propose a concrete greedy sequencing rule, which structurally inhibits the feasibility of sandwich attacks. Chan, Wu, and Shi~\cite{chan2024mechanism} initiate the mechanism design approach/philosophy for mitigating MEV and study a model where relevant strategic players (user or miner) aim to obtain \emph{risk-free profit}, whereas we focus on non-atomic arbitrage. AnimaguSwap~\cite{wadhwa2023data} presents an AMM design to achieve data-independent ordering of user transactions at the application level. V0LVER~\cite{mcmenamin2023amm} presents an AMM against LVR and MEV based on an encrypted transaction mempool.

\parhead{Other MEV mitigation solutions}
Mitigating the negative effects of MEV is a research focus in both academia and industry, with various studies conducted from different perspectives.
A commonly used practice for MEV mitigation is using private channels~\cite{MEVBlocker,flashbots2024protect}, where user transactions are sent directly to block builders, bypassing the public mempool and thus preventing MEV attacks like front-running and sandwich attacks. 
Another approach focuses on ensuring time-based order fairness~\cite{kelkar2020order, kelkar2023themis, zhang2020byzantine, cachin2022quick}. Miners or validators adhering to time-based order fairness must order transactions based on the time they receive the transaction, preventing MEV caused by ex-post order manipulation. 
A different method to achieve fair ordering involves users first committing to their transactions and revealing them only after the order has been determined. Thus, the transaction order is determined independently of the content. 
We refer readers to an SoK paper~\cite{yang2022sok} for a comprehensive understanding of various MEV mitigation approaches.

\parhead{Sybil-proofness in mechanism design in blockchain.} The permissionless nature of blockchain makes it very easy for participants to create multiple identities, thus launching Sybil attacks. To see the challenges, the most successful mechanisms in classic auction theory, such as second-price auctions and VCG auctions, are vulnerable to Sybil attacks. Since the pioneered work by Roughgarden on transaction fee mechanism design~\cite{DBLP:conf/sigecom/Roughgarden21} and also for practical consideration, Sybil-proofness has become one of the most desired properties for mechanism design in blockchain. Notable examples include the rich transaction fee mechanism design literature~\cite{DBLP:conf/sigecom/Roughgarden21,DBLP:conf/aft/BahraniGR24,chung2024collusion,gafni2024barriers}, voting mechanism~\cite{lenzi2024efficient}, and the very recent mechanism in ZK-Rollup prover markets~\cite{wang2024mechanism}. Mazorra and Penna~\cite{mazorra2023towards} consider a similar MEV rebates problem in the context of MEV-share combinatorial order flow auctions. They discuss the Shapley value-based mechanism and show that in their setting, any symmetric, efficient, and Strong monotonic rebate mechanism is weak against Sybil strategies.

\section{Model}\label{sec:model}

\subsection{The Basic Model}\label{subsec:basic}
This section describes the basic CFMM pool and all involved players, including liquidity providers, users, and arbitrageurs.

\parhead{CFMM Pool}\label{subsec:CFMM}
We consider a CFMM pool with a risky asset $\mathcal{X}$ (e.g., ETH) and a num\'{e}raire asset $\mathcal{Y}$ (e.g., USDC). The pool state $s=(x,y)$ is specified by the current reserves of tokens and should satisfy a constant function $F(x,y)$. The slope at a state $(x,y)\in F$ is the spot price or marginal exchange rate $\left|\frac{\partial F/\partial x}{\partial F/\partial y}\right|$. 

We do not prescribe a constant function but only require two natural properties on $F$: (1) for any two points $(x,y),(x', y')\in F$, we have $x>x'\Leftrightarrow y<y'$, namely, when the reserve of $\calX$ increases, the reserve of $\calY$ decreases and vice versa; (2) $F(x,y)$ is differentiable and the marginal exchange rate $\left|\frac{\partial F/\partial x}{\partial F/\partial y}\right|$ is decreasing with respect to $x$. Most CFMMs satisfy these two properties, including Uniswap v2 and v3. The two properties imply that for any $x$, there is exactly one $y$ such that $(x,y)\in F$ and vice versa. To simplify notations, we use $F_y(x)$ to denote that $y$ such that $(x,y)\in F$ and similarly define $F_x(y)$.

We call the pool's state after the latest block its \textit{initial state}  $s_0 = (x_0, y_0)$.

\parhead{Liquidity Providers} 
Liquidity providers contribute pairs of tokens (e.g., ETH and USDC) to a CFMM pool, which is then used to fulfill users' trade orders. In exchange, LPs earn fees from each trade in the pool. A small fraction $f\in[0,1)$ of each trade's input tokens are charged as swap fees and are shared by liquidity providers proportional to their contribution to liquidity reserves. 
The fees are temporarily kept by the pool and can be withdrawn by liquidity providers by burning their liquidity tokens. 
In our design, the swap fees are stored separately like in Uniswap v3, rather than being continuously deposited in the pool as liquidity (e.g., Uniswap v2). 
Throughout this paper, we only consider swap-like transactions and assume no liquidity deposit or redemption.

\parhead{Users/Noise Traders} Users are a population of noise traders who intend to trade through the CFMM pool. Each user holds one type of asset and seeks to exchange it for another by creating a swap transaction. 
Specifically, a transaction $\TX=(\calX \to \calY, \delta_{\calX}^{in}, \delta_{\calY}^{out})$ specifies that the user is willing to spend up to $\delta_{\calX}^{in}$ units of asset $\calX$, provided they receive at least $\delta_{\calY}^{out}$ units of asset $\calY$. Similarly, the signer of transaction $\TX=(\calY \to \calX, \delta_{\calY}^{in}, \delta_{\calX}^{out})$ allows at most $\delta_{\calY}^{in}$ units of $\calY$ to be taken from their account if they receive at least $\delta_{\calX}^{out}$ units of $\calX$. For simplicity of notations, the terms $\delta_{\calX}^{in}$ and $\delta_{\calY}^{in}$ refer to the actual quantities of input tokens available for trading, after the fees have been deducted.
In this paper, we use trade/order/transaction interchangeably. 

\parhead{Arbitrageurs/Informed Traders} Arbitrageurs are informed traders with access to external markets. They continuously monitor the price disparities between the CFMM and external markets, seeking to profit by arbitrage. %

Each arbitrageur $i$ holds a private belief $v_i^*$ regarding the external price of the risky asset $\calX$. These beliefs decide when and how to execute arbitrage opportunities. 
Arbitrageurs may have different beliefs for various reasons. First, they may focus on different external markets or operate on different timelines for executing off-chain trades. Second, prices can exhibit significant volatility, leading to varying price expectations among arbitrageurs. Lastly, arbitrageurs' trading costs may be different, especially in off-chain trading venues.

\parhead{Remarks}
In today's MEV supply chains, when arbitrageurs compete for such MEV opportunities through MEV auctions~\cite{daian2020flash}, most MEV is captured by block builders and proposers. 
This is undesirable for redistribution and is avoided in \AMMName because arbitrageurs' competition happens on the CFMM side. 
A notable feature of the MEV-redistribution mechanism is its ability to insert MEV transactions on behalf of arbitrageurs,
which do not need to pay swap fees (namely, $f=0$ for them). This enables arbitrageurs to correct even tiny price
discrepancies~\cite{adams2024amm} and makes it possible to compute the optimal MEV strategy efficiently (otherwise, the problem becomes NP-hard~\cite{zhang2024computation}).

\subsection{MEV-Redistribution Mechanism}
The key novelty of \AMMName is a new MEV-redistribution mechanism, which decides which user transactions should be included in the bundle, which arbitrageurs have the right to insert arbitrage transactions, the sequence of these transactions, and how much MEV value should be refunded to users and liquidity providers. These decisions are formalized by three functions: a bundle generation rule, a payment rule, and a refund rule.

Before presenting the formal definitions, we first introduce necessary notations.
We use $N$ to denote the set of $n$ arbitrageurs, and $\mathbf{q}$ to denote a $n$-sized vector where $q_i$ is arbitrageur $i$'s report on the external price of the risky asset $\calX$. 
Their true belief is $v_i^*$, which may differ from $q_i$.
We assume each arbitrageur $i$ can create any number of Sybil transactions, the set of which is denoted by $S_i$. Thus, we use $M = R \cup S$ to denote the set of pending transactions in the CFMM pool, where $R$ is a set of \textit{real} transactions from users, and $S = S_1 \cup S_2 \cup \cdots \cup S_n$.

\begin{definition}[Bundle Generation Rule]
    A \textit{bundle generation rule} is a function $\mathbf{b}$ from the pool's initial state $s_0$, pending transactions $M$, and arbitrageurs' report $\mathbf{q}$ to an ordered set $B$ of transactions (a bundle). The elements in $B$ are $T\cup A$, where $T\subseteq M$ is a subset of pending transactions, and $A = A_1 \cup \cdots \cup A_n$ is the set of MEV transactions inserted by the rule on behalf of arbitrageurs. 
\end{definition} 

\begin{definition}[Payment Rule]
    A \textit{payment rule} is a function $\mathbf{p}$ from the pool's initial state $s_0$, pending transactions $M$, and arbitrageurs' report $\mathbf{q}$ to a non-negative number $p_i(s_0, M, \mathbf{q})$ for each arbitrageur $i\in [n]$.
\end{definition}

\begin{definition}[Refund Rule]
    A \textit{refund rule} is a function $\mathbf{r}$ from the pool's initial state $s_0$, pending transactions $M$, and arbitrageurs' report $\mathbf{q}$ to a non-negative number $r(s_0, M, \mathbf{q}, \TX_j)$ for each transaction $\TX_j \in M$; the remaining payments from arbitrageurs, i.e., $\sum_{i\in [n]} p_i(s_0, M, \mathbf{q}) - \sum_{\TX_j \in M} r(s_0, M, \mathbf{q}, \TX_j)$ are refunded to liquidity providers.
\end{definition}

When the context is clear, we may shorten $r(s_0, M, \mathbf{q}, \TX_j)$ to $r(\TX_j)$ or $r_j$.

\begin{definition}[MEV-Redistribution Mechanism]
    An \textit{MEV-redistribution mechanism} is a triple $(\mathbf{b}, \mathbf{p}, \mathbf{r})$ in which $\mathbf{b}$ is a bundle generation rule, $\mathbf{p}$ is a payment rule, and $\mathbf{r}$ is a refund rule.
\end{definition}

\subsection{Desired Properties}

This section formalizes what it means for an MEV-redistribution mechanism to be game-theoretically sound. 
First of all, an arbitrageur should be incentivized to report their true belief in the external price of the risky asset $\calX$ (defined as ``Truthful'' below). 
Meanwhile, we must prevent arbitrageurs from stealing user refunds by posing as users in Sybil attacks.
As a reminder, we assume that arbitrageurs can submit any number of Sybil transactions, which may receive some refunds. 
We require that inserting Sybil transactions will not harm users' utilities (defined as ``Sybil-proof'' below).

\begin{definition}[Arbitrageur Utility Function]\label{def:attacker_utility}
     For an MEV-redistribution mechanism $(\mathbf{b}, \mathbf{p}, \mathbf{r})$, pool's initial state $s_0$, pending transactions $M$, arbitrageurs' report $\mathbf{q}$, and generated bundle $B$, the utility of arbitrageur $i$ with true belief $v_i^*$ and Sybil transactions $S_i \subseteq M$ is 
    \begin{subequations}   
        \begin{align}
        u\big(S_i, q_i; S_{-i}, \mathbf{q}_{-i}\big) \coloneqq & \quad \sum_{\TX_{j} \in S_i\cap B}
        \Big[ \left(x_{j-1}-x_j\right) \cdot v_i^* + (y_{j-1}-y_j) + r\left(\TX_{j}\right) \Big] \label{eq:sybiltxs}\\
        &+ \sum_{\TX_{j} \in A_i}
        \Big[ \left(x_{j-1}-x_j\right) \cdot v_i^* + (y_{j-1}-y_j) \Big] \label{eq:arbtxs} \\
        &- \quad p_i,   \label{eq:payment}
        \end{align}
        \label{eq:attacker_utility}
    \end{subequations}
     where $A_i \subseteq B$ is $i$'s arbitrage transactions inserted by the bundle generation rule $\mathbf{b}$ and $(x_j, y_j)$ is the pool's state after the $j$-th transaction in the bundle.
\end{definition}

On the LHS of (\ref{eq:attacker_utility}), we highlight the dependence
of the utility function on the two arguments that are under arbitrageur $i$'s direct control: the choices of
Sybil transactions $S_i$ and the report on external price $q_i$. 
We also explicitly show its dependence on other arbitrageurs' strategy $S_{-i}$ and $\mathbf{q}_{-i}$. 
Here, $\mathbf{q}_{-i}$ means the vector $\mathbf{q}$ of all reports, but with the $i$-th component removed. This is a standard notation in economics, and we will also use similar notations for other objects (e.g., $S_{-i}$ above). 
The formula (\ref{eq:sybiltxs}) sums over $i$'s Sybil transactions included in the generated bundle, the first two terms of which represent $i$'s revenue from transaction execution, and the third term is the refund received by this transaction. The formula (\ref{eq:arbtxs}) sums over the revenue from $i$'s arbitrage transactions added in the bundle by the mechanism. The formula (\ref{eq:payment}) is the cost $i$ needs to pay for the MEV opportunity, the output of the payment rule.

\begin{definition}[Truthful] 
    An MEV-redistribution mechanism $(\mathbf{b}, \mathbf{p}, \mathbf{r})$ is truthful if, for every pool's initial state $s_0$, pending transactions $M = R \cup S_{-i}$, and other arbitrageurs' report $\mathbf{q}_{-i}$, for every  arbitrageur $i$, it  maximizes its utility (\ref{eq:attacker_utility}) by reporting its true belief (i.e., setting $q_i = v_i^*$), namely,

    \begin{center}
        $u(S_i=\emptyset,v_i^*;S_{-i},\mathbf{q}_{-i})\geq u(S_i=\emptyset,q_i;S_{-i},\mathbf{q}_{-i})$ for all $q_i\in \mathbb{R}$.
    \end{center}
    
\end{definition}

Note that we separate the discussion of truthfulness from Sybil proofness, and the above definition assumes that arbitrageur $i$ does not add Sybil transactions.

As mentioned above, 
arbitrageurs may steal refunds by pretending to be users and submitting fake (Sybil) transactions.
An MEV-redistribution mechanism is \textit{Sybil-proof} in that creating Sybil transactions does not reduce any user's utility.

\begin{definition}[User Utility Function]\label{def:userutility}
    For an MEV-redistribution mechanism $(\mathbf{b}, \mathbf{p}, \mathbf{r})$, pool's initial state $s_0$, pending transactions $M=R\cup S$, arbitrageurs' report $\mathbf{q}$, the utility of the originator of a transaction $\TX_j \in R$ is 
    \begin{equation}
        u(\TX_j \mid S,\mathbf{q}) = \frac{\delta_{\calY}}{\delta_{\calX}} \cdot \calX_j + \calY_j,
        \label{eq:userutility}
    \end{equation}
    where $(\calX_j, \calY_j)$ is the number of tokens the originator has after executed through the mechanism, and $\delta_{\calX} = \delta_{\calX}^{in}$, $\delta_{\calY} = \delta_{\calY}^{out}$ if $\TX=(\calX \to \calY, \delta_{\calX}^{in}, \delta_{\calY}^{out})$; $\delta_{\calX} = \delta_{\calX}^{out}$, $\delta_{\calY} = \delta_{\calY}^{in}$ if $\TX=(\calY \to \calX, \delta_{\calY}^{in}, \delta_{\calX}^{out})$.
\end{definition}

Initially, it's assumed that the user of an $\calX \to \calY$ transaction holds $\delta_{\calX}^{in}$ units of $\calX$ token and no $\calY$ token, while the user of an $\calY \to \calX$ transaction has $\delta_{\calY}^{in}$ units of $\calY$ token and no $\calX$ token. After execution, the user's utility depends on the results of the mechanism, where $\calX_j$ and $\calY_j$ in \Cref{eq:userutility} counts not only the direct execution of $\TX_j$ but also the extra refund it receives. Additionally, we highlight the dependence of the utility function on the (partial) inputs $S$ and $\mathbf{q}$ of the mechanism, which are under arbitrageurs' control though.

\begin{definition}[Sybil-proof]
    An MEV-redistribution mechanism $(\mathbf{b}, \mathbf{p}, \mathbf{r})$ is Sybil-proof if, for any initial state $s_0$, pending transactions $M = R \cup S_{-i}$, and arbitrageurs' report $\mathbf{q}$, creating Sybil transactions $S_i$ will not reduce users' utility:
    \begin{center}
        $u(\TX_j \mid \emptyset \cup S_{-i}, \mathbf{q})\geq u(\TX_j \mid S_i \cup S_{-i}, \mathbf{q})$ for all $\TX_j\in R$ and all $S_i$.
    \end{center}
    
\end{definition}

\section{Optimal MEV with Public Orders}\label{sec:optimalMEV}
Before delving into \AMMName, it is helpful to take a slight detour to consider an optimization problem in standard AMMs where user trades are \textit{public}. (To reiterate, \AMMName hides user transactions from the public, including arbitrageurs.)

In this setting, transactions are visible to everyone, so arbitrageurs can create their own MEV bundles to exploit arbitrage opportunities between the CFMM and external markets (assuming no swap fee). The key question for each arbitrageur is: what is the optimal MEV strategy to maximize utility given the public order flows?%

Without loss of generality, we can discuss this problem from the perspective of an arbitrary arbitrageur, whose belief in the external price of the risky asset $\calX$ is assumed to be $v^*$. Note that for an arbitrary $v^* >0$, there is exactly one pool state $s^* = (x^*, y^*)$ at which the pool's marginal exchange rate $\left|\frac{\partial F/\partial x}{\partial F/\partial y}\right| = v^*$. So we use $v^*$ and $(x^*, y^*)$ interchangeably to represent the arbitrageur's belief. Also note that arbitrageurs do not need to create Sybil transactions in this setting, as they can insert arbitrary transactions when constructing the bundle.

\subsection{Optimal MEV Strategy}\label{subsec:optimalMEV}

\begin{algorithm}[!ht]
	\caption{Optimal MEV Strategy under Public Order}
	\label{alg:optimalMEV}
	\KwIn{An initial state $s_0=(x_0,y_0)$, a set of users transactions $\set{\TX_j}_{j\in[m]}$, and the arbitrageur's belief $v^*$ which corresponds to the state $s^* = (x^*, y^*)$.
	}
	\KwOut{A bundle for the arbitrageur to obtain the maximal MEV.
	}
	
	\BlankLine
	
	\BlankLine

    \For{each $j \in [1:m]$}{ \label{algline:enumerate}
    Current state $s_{j-1} = (x_{j-1}, y_{j-1})$.{\hfill\tcp{Token reserves}}

    \If{$\TX_j = (\calX \to \calY, \delta_{\calX}^{in}, \delta_{\calY}^{out})$}{\label{line:preprocessBegin}
        $(x_j^\ell, y_j^\ell)$ is the limit state satisfying $y_j^\ell - F_{y}(x_j^\ell + \delta_{\calX}^{in}) = \delta_{\calY}^{out}$.
        
        Let $\Delta x \gets \delta_{\calX}^{in}$, $\Delta y \gets -\delta_{\calY}^{out}$.

    }
    \ElseIf{$\TX_j=(\calY \to \calX, \delta_{\calY}^{in}, \delta_{\calX}^{out})$}{
        $(x_j^\ell, y_j^\ell)$ is the limit state satisfying $x_j^\ell - F_x(y_j^\ell + \delta_{\calY}^{in}) = \delta_{\calX}^{out}$.
    
        Let $\Delta x \gets -\delta_{\calX}^{out}$, $\Delta y \gets \delta_{\calY}^{in}$. 
    }\label{line:preprocessEnd}

    Let $\Delta \phi \gets \Delta x \cdot v^* + \Delta y$.   
    
    \If{$\Delta \phi \geq 0$}{   \label{line:insertBegin}
        \If{$x_{j-1} < x_j^\ell$}{
            Insert a transaction $\TX =\Big( \calX \to \calY, \delta_{\calX}^{in} = x_j^\ell-x_{j-1}, \delta_{\calY}^{out} = y_{j-1}-y_j^\ell \Big)$.
        }
        \ElseIf{$x_{j-1} > x_j^\ell$}{
            Insert a transaction $\TX =\Big(\calY \to \calX, \delta_{\calY}^{in}=y_j^\ell-y_{j-1}, \delta_{\calX}^{out}=x_{j-1} - x_j^\ell \Big)$.
        }
    
        Place the user transaction $\TX_j$.
    }   \label{line:insertEnd}
    
    }\label{algline:enumerateEnd}
    Suppose the current state is $(x, y)$. \label{algline:arbitrage}

    \If{$x < x^*$}{
        Add a transaction $\TX = \Big( \calX \to \calY, \delta_{\calX}^{in}=x^*-x, \delta_{\calY}^{out}=y-y^* \Big)$.
    }
    \ElseIf{$x > x^*$}{
        Add a transaction $\TX =\Big(\calY \to \calX, \delta_{\calY}^{in}=y^* - y,  \delta_{\calX}^{out}=x - x^* \Big)$.
    }\label{algline:arbitrageEnd}
\end{algorithm}

Algorithm~\ref{alg:optimalMEV} presents an efficient strategy to extract the maximal MEV. A formal definition of the MEV optimization problem and a detailed elaboration of Algorithm~\ref{alg:optimalMEV} are shown in~\cref{sec:optimal}, in the interest of space.  
Our optimal MEV strategy generalizes the polynomial-time algorithm for constant product AMM by Bartoletti et al.~\cite{bartoletti2022maximizing} to work for all CFMMs satisfying the two natural properties defined in \cref{subsec:basic}. 
Below, we briefly present the main idea, focusing on the new perspective gained by interpreting the optimal MEV strategy using \textit{potential function}.

In a special case where there is no user transaction, it is not hard to show that the arbitrageur's best strategy is to insert an arbitrage transaction to change the state $(x, y)$ to $(x^*, y^*)$ at which $\left|\frac{\partial F/\partial x}{\partial F/\partial y}\right| = v^*$. Given $v^*$, the profit from arbitrage only depends on the starting point $(x, y)$. In other words, each state $(x, y)$ corresponds to an arbitrage profit, which can be understood as the potential value of that state. So, for the arbitrageur with belief $v^*$, we define the \textit{potential value} of any pool state $s=(x,y)$ as
\begin{equation}\label{eq:potential}
    \phi(s, v^*) = (x \cdot v^* + y) - (x^* \cdot v^* + y^*).
\end{equation}
When the context is clear, we abbreviate this as $\phi(s)$.

When there is a set of $m$ pending user transactions, the key observation is that the execution of certain user transactions can increase the potential profits of arbitrage. Intuitively, the user of $\TX_j$ and arbitrageur form a zero-sum game, so the worst execution of $\TX_j$ is the best for the arbitrageur, which happens if $\TX_j$ is executed at its limit state\footnote{$\TX_j$'s \textit{limit state} $s_j^\ell = (x_j^\ell, y_j^\ell)$ is defined as the state at which, when executed, the transaction will pay exactly the maximum amount of input token and receive the minimum amount of output token.}: $(x_j^l, y_j^l) \xrightarrow[]{\TX_j} (x_j^l + \Delta x_j, y_j^l + \Delta y_j)$, where $(\Delta x_j, \Delta y_j)$ is $\TX_j$'s impact on the trading pool when executed at its limit state, namely,
\begin{equation}  
    (\Delta x_j, \Delta y_j) \coloneqq
    \left\{  
        \begin{array}{lr}  
            (\delta_{\calX}^{in}, -\delta_{\calY}^{out}), & \text{ ~~when } \TX_j = (\calX \to \calY, \delta_{\calX}^{in}, \delta_{\calY}^{out}); \\  
            (-\delta_{\calX}^{out}, \delta_{\calY}^{in}), & \text{ ~~when } \TX_j=(\calY \to \calX, \delta_{\calY}^{in}, \delta_{\calX}^{out}).    
        \end{array}  
    \right.
    \label{eq:impact}
\end{equation}
From the perspective of the arbitrageur, whether a transaction $\TX_j$ should be executed depends on its impact on the potential value of the pool state, namely, the difference between the potential value of the post-execution state and that of the pre-execution state. We define such difference as $\TX_j$'s \textit{potential value} $\Delta\phi_j$:
\begin{equation}\label{eq:deltaPhi}
    \Delta \phi_j \coloneqq \phi(x_j^l + \Delta x_j, y_j^l + \Delta y_j) - \phi(x_j^l, y_j^l) = \Delta x_j \cdot v^* + \Delta y_j,
\end{equation}
where $v^*$ is the arbitrageur's belief in the external price. If $\Delta\phi_j \geq 0$, $\TX_j$ is profitable, then it will be placed in the bundle and executed at its limit state; otherwise, it will be ignored and bring zero profit. Combining these two cases, the \textit{actual} value of an arbitrary transaction $\TX_j$ is denoted by
\begin{equation}
    V(\TX_j) = \max\{0, \Delta\phi_j\}.
    \label{eq:TXvalue}
\end{equation}

Based on the above intuition, we reach the following result:

\begin{theorem}\label{theorem:publicOF} 
The arbitrageur's maximal MEV is $\phi(s_0, v^*) + \sum_{j\in[m]} V(\TX_j)$. Furthermore, Algorithm~\ref{alg:optimalMEV} can construct the bundle that obtains the maximal MEV in polynomial time. 
\end{theorem}

The proof is non-trivial and can be found in~\cref{subsec:mevproof}.

In simple terms, Algorithm~\ref{alg:optimalMEV} enumerates each pending transaction $\TX_j$ ($j\in[m]$), inserts a frontrunning transaction to make $\TX_j$ execute at its limit state if it has non-negative potential value (i.e., $\Delta\phi_j \geq 0$), and at the end, concludes with a single backrunning transaction to capture all arbitrage profits and stop at the no-arbitrage state $(x^*, y^*)$ aligned with the arbitrageur's belief $v^*$.  
To provide intuition, we use~\Cref{ex:optimalMEV} to demonstrate how Algorithm~\ref{alg:optimalMEV} operates and illustrate its capability to capture the maximal MEV.
\begin{example}\label{ex:optimalMEV}
    Consider a trading curve defined by $F(x,y) = xy = 400$, with an initial state of $s_0 = (4,100)$ and a swap fee $f=0$. The arbitrageur holds a belief about the external price of asset $\calX$, valuing it at $v^* = 4$, which corresponds to the no-arbitrage state $(10, 40)$. There are three pending transactions: $\TX_1=(\calX \to \calY, \delta_{\calX}^{in}=8, \delta_{\calY}^{out}=25)$, $\TX_2=(\calX \to \calY, \delta_{\calX}^{in}=30, \delta_{\calY}^{out}=12)$, and $\TX_3=(\calY \to \calX, \delta_{\calY}^{in}=20, \delta_{\calX}^{out}=10)$. 
    
    We now compute the potential value of the initial state, as well as each transaction's limit state, its impact on the trading pool, potential value, and actual value:
    \begin{itemize}
        \item Initial state $s_0$: $\phi(s_0, v^*) = 36$;

        \item $\TX_1$: $(x_1^\ell, y_1^\ell) = (8,50)$, $(\Delta x_1, \Delta y_1) = (8, -25)$, $\Delta \phi_1 = 7$, $V(\TX_1) = 7$;
        
        \item $\TX_2$: $(x_2^\ell, y_2^\ell) = (20,20)$, $(\Delta x_2, \Delta y_2) = (30, -12)$, $\Delta \phi_2 = 108$, $V(\TX_2) = 108$;

        \item $\TX_3$: $(x_3^\ell, y_3^\ell) = (20,20)$, $(\Delta x_3, \Delta y_3) = (-10, 20)$, $\Delta \phi_3 = -20$, $V(\TX_3) = 0$.
    \end{itemize}
    
    According to~\Cref{theorem:publicOF}, the arbitrageur's maximal MEV is $\phi(s_0, v^*) + \sum_{j\in[1:3]} V(\TX_j) = 151$. 
    To illustrate the operations and correctness of Algorithm~\ref{alg:optimalMEV}, we first provide a more intuitive (but less efficient) way to extract the maximal MEV, and then show its simplification by steps, which corresponds to the result of Algorithm~\ref{alg:optimalMEV}. 
    
    The more intuitive way is to first ``sandwich'' $\TX_1$ and $\TX_2$ (by frontrunning each transaction so it executes at its limit state, followed by a backrunning transaction to revert to the initial state), and then arbitrage the pool to eventually reach the no-arbitrage state $(10, 40)$, as illustrated in~\Cref{fig:exp1-sandwich}. To calculate the arbitrageur's profit from this bundle, we can split each backrunning transaction into two smaller ones, allowing the revenue from one of them to offset that from the frontrunning transaction. For instance, $\TX_1^{\textsf{Back}}$ can be divided into $\TX_1^{\textsf{Back1}}$ and $\TX_1^{\textsf{Back2}}$, changing the state as follows: $(16, 25) \xrightarrow[]{\TX_1^{\textsf{Back1}}} (8, 50) \xrightarrow[]{\TX_1^{\textsf{Back2}}} (4, 100)$. In this way, the revenue from $\TX_1^{\textsf{Back2}}$ cancels out the revenue from $\TX_1^{\textsf{Front}}$, making it easy to verify that the profit from the first ``sandwich'' (i.e., $\TX_1^{\textsf{Front}}$, $\TX_1$, and $\TX_1^{\textsf{Back}}$) is exactly $V(\TX_1)$. Similarly, the profit from the second ``sandwich'' (i.e., $\TX_2^{\textsf{Front}}$, $\TX_2$, and $\TX_2^{\textsf{Back}}$) is $V(\TX_2)$. Finally, the profit from the last arbitrage transaction $\TX_0^{\textsf{Arb}}$ is $-(10-4)\times 4 + (100-40) = \phi(s_0, v^*)$. Hence, the bundle in~\Cref{fig:exp1-sandwich} extracts the maximal MEV for the arbitrageur. 

    Another approach to split the MEV transitions is shown in~\Cref{fig:exp1-split}, where $\TX_2^{\textsf{Front}}$ and $\TX_2^{\textsf{Back}}$ are divided so that the total profit from transactions with the same symbol is 0. Removing the four transactions with symbols leaves a simplified bundle, shown in~\Cref{fig:exp1-bundle}, which is the outcome of Algorithm~\ref{alg:optimalMEV}. This bundle extracts the same profits as the previous one, confirming that Algorithm~\ref{alg:optimalMEV} can capture the maximal MEV. 
\end{example}

\begin{figure}[htbp]
    \centering
    \begin{subfigure}{0.95\textwidth}
        \centering
        \includegraphics[width=\linewidth]{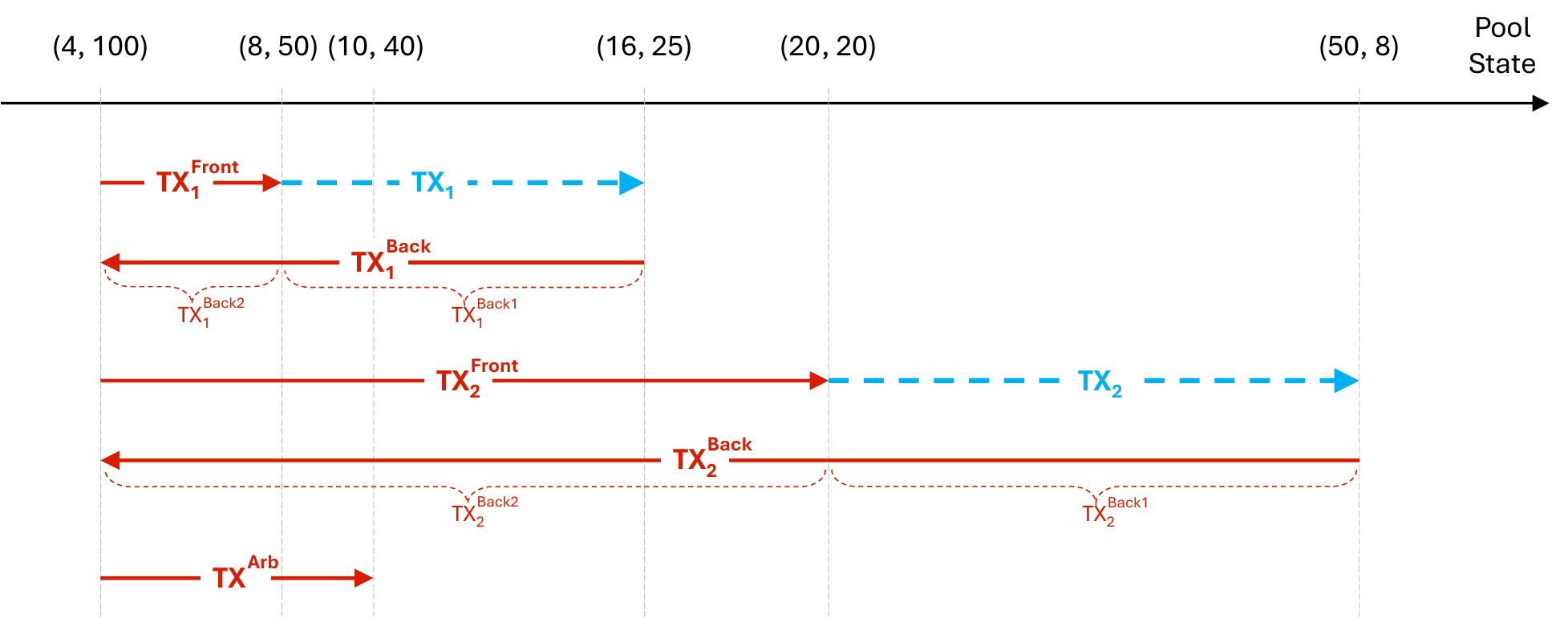}
        \caption{A bundle to extract the maximal MEV by ``sandwiching'' each pending transaction with a non-negative potential value, followed by an arbitrage transaction to reach the no-arbitrage state. $\TX_{1(2)}^{\textsf{Back}}$ can be divided into $\TX_{1(2)}^{\textsf{Back1}}$ and $\TX_{1(2)}^{\textsf{Back2}}$ and the total profit from $\TX_{1(2)}^{\textsf{Back2}}$ and $\TX_{1(2)}^{\textsf{Front}}$ is 0.}
        \label{fig:exp1-sandwich}
    \end{subfigure}

    \vspace{0.5cm}
    
    \begin{subfigure}{0.95\textwidth}
        \centering
        \includegraphics[width=\linewidth]{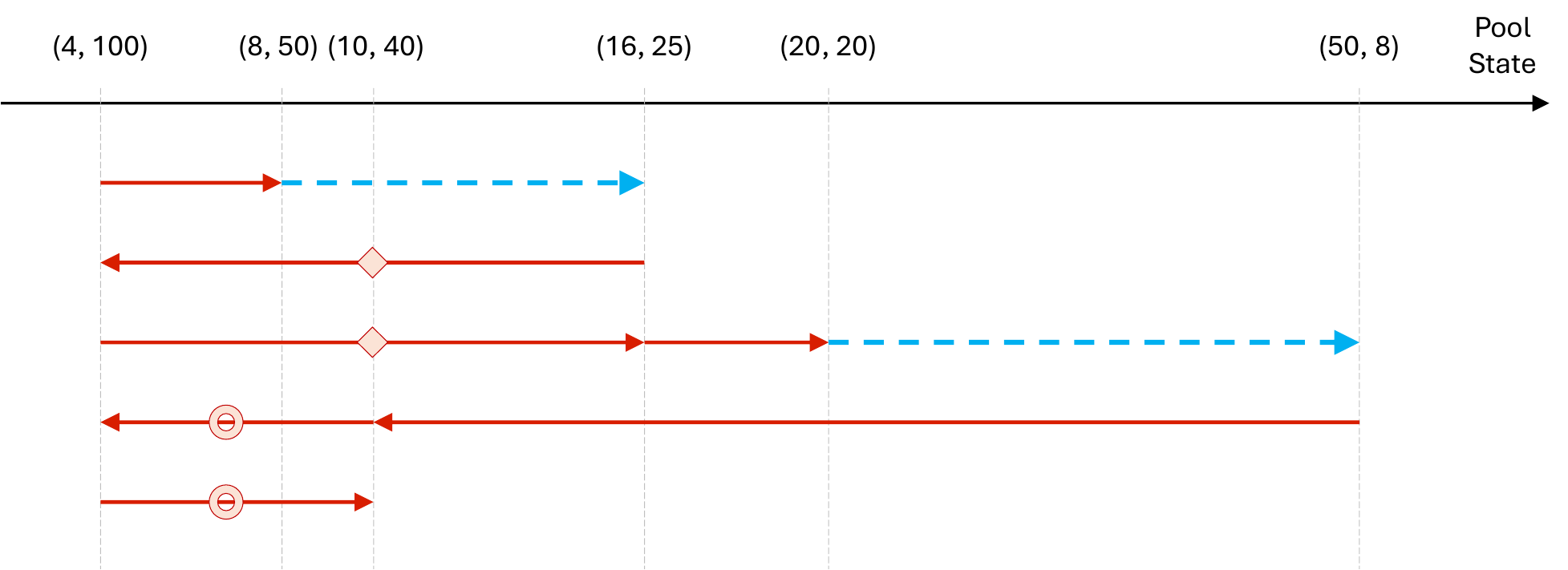}
        \caption{Splitting $\TX_2^{\textsf{Front}}$ and $\TX_2^{\textsf{Back}}$ into two smaller ones. The total profit from transactions marked with the same symbol (\scalebox{1.3}{$\diamond$} or ${\circledcirc}$) is 0.}
        \label{fig:exp1-split}
    \end{subfigure}
    
    \vspace{0.5cm}
    
    \begin{subfigure}{0.95\textwidth}
        \centering
        \includegraphics[width=\linewidth]{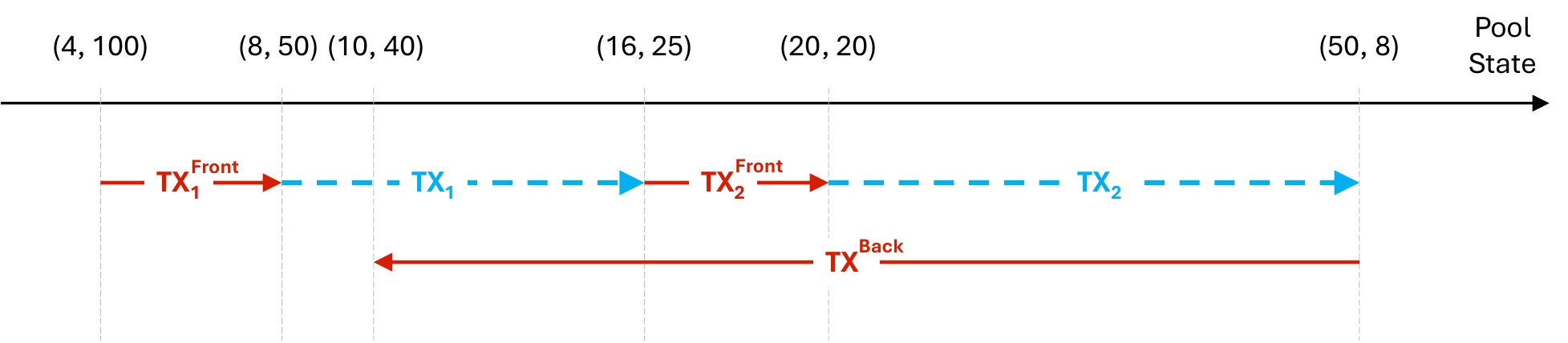}
        \caption{Removing the symbol-marked transactions from (b) results in a simplified bundle, which corresponds to the outcome of Algorithm~\ref{alg:optimalMEV} and captures the maximal MEV.}
        \label{fig:exp1-bundle}
    \end{subfigure}
    \caption{An illustration of the optimal MEV strategy for~\Cref{ex:optimalMEV}. Blue dotted lines represent pending transactions from users, while red solid lines represent MEV transactions from the arbitrageur. The transactions in each subgraph form an MEV bundle. Pool states $(x,y)$ in the figure are shown in ascending order based on the value of $x$. }
    \label{fig:alg1}
\end{figure}

\subsection{Optimal MEV = \textsf{LVR} + Each User's Loss}\label{sec:MEVnLVR}

The maximal MEV value, namely, the arbitrageur's total profit from Algorithm~\ref{alg:optimalMEV}, represents the total loss of liquidity providers and users. One key observation in \AMMName is to attribute MEV to the loss of LP and each user as follows:

\begin{equation}
    U = \underbrace{\phi(s_0)}_{\text{LPs' loss \textsf{LVR}}} + \underbrace{V(\TX_1)}_{\text{User $1$'s loss}} + \underbrace{V(\TX_2)}_{\text{User $2$'s loss}} + \cdots +
    \underbrace{V(\TX_m)}_{\text{User $m$'s loss}}.
    \label{eq:MEVnLVR}
\end{equation}

\vspace{0.8em}
\parhead{Justification of $\phi(s_0)$ = \textsf{LVR}}
LVR~\cite{milionis2022automated} quantifies the costs incurred by LPs due to \textit{rebalancing arbitrage}, a process that corrects the pool's stale price when the external market price changes. LVR is based on a \textit{rebalancing strategy}, which manages an off-chain portfolio of tokens $\calX$ and $\calY$ to mirror the CFMM pool's holdings of the risky asset $\calX$ at all times. Specifically, whenever an arbitrage transaction occurs in the pool in response to external price movements---causing the pool to either buy or sell the risky asset---the rebalancing strategy replicates the pool's sell/buy behavior in the off-chain portfolio at the external market using the CEX price. The profit from this rebalancing strategy, representing the difference in monetary value between the LPs' portfolio in the AMM pool and the off-chain rebalancing portfolio, is what defines LPs' \textsf{LVR}. More concretely, let $\Delta$ denote the number of risky assets \textit{sold} by the pool due to the rebalancing arbitrage. Let $p^{\textsf{CEX}}$ represent the external price and $p^{\textsf{AMM}}$ the execution price of the arbitrage transaction in the pool. LPs' \textsf{LVR} from this trade is then calculated as $\Delta \cdot (p^{\textsf{CEX}} - p^{\textsf{AMM}})$, which is always positive\footnote{When the pool sells risky assets, $\Delta > 0$ and $p^{\textsf{CEX}} > p^{\textsf{AMM}}$; when the pool buys risky assets, $\Delta < 0$ and $p^{\textsf{CEX}} < p^{\textsf{AMM}}$.}.

Note that in addition to \textit{rebalancing arbitrage}, the authors of LVR also mention the \textit{reversion arbitrage}, which occurs when user transactions cause the DEX price to deviate from the CEX price, creating arbitrage opportunities. At a high level, LVR focuses on rebalancing arbitrage, whereas the maximal MEV value above considers both rebalancing and reversion arbitrage. Despite this distinction, it is natural to apply the rebalancing strategy in our context to compute LPs' loss within the bundle of Algorithm~\ref{alg:optimalMEV}. Specifically, we can apply the rebalancing strategy described above to all transactions in the bundle and compute the cumulative profits, which corresponds to \textsf{LVR}. Suppose there are $k$ transactions in the bundle and the $j$-th trade alters the pool state from $(x_{j-1}, y_{j-1})$ to $(x_j, y_j)$, with the end state being $(x_k, y_k) = (x^*, y^*)$. Then, by definition, we have: 
\begin{small}
    \begin{equation*}
    \mathsf{LVR} = \sum_{j\in [1:k]} (x_{j-1} - x_j) \cdot (v^* - \frac{y_j - y_{j-1}}{x_{j-1} - x_j}) = \sum_{j\in [1:k]} \Big[(x_{j-1} - x_j) \cdot v^* - (y_j - y_{j-1}) \Big] = \phi(s_0).
\end{equation*}
\end{small}

In the paper that introduces LVR~\cite{milionis2022automated}, the authors state that ``Our model allows us to quantify the magnitude of profits of rebalancing arbitrageurs, but not reversion arbitrageurs.'' The above analysis provides a perspective to understand this statement. It does not mean the rebalancing strategy is not applicable to scenarios involving reversion arbitrage. Rather, when taking the strategy to replicate both noise trades and informed trades without distinction, the profit of each step can be positive or negative. Ultimately, the cumulative profits from noise trades and reversion arbitrage will sum to zero, leaving only the profits derived from rebalancing arbitrage.

\vspace{0.5em}
\parhead{User loss}

Recall that $V(\TX_j) = \max\{0, \Delta\phi_j\}$ represents the \textit{actual} value of transaction $\TX_j$ to the arbitrageur, where $\Delta\phi_j$ represents $\TX_j$'s maximal \textit{potential} arbitrage value. If $\Delta\phi_j \geq 0$, executing $\TX_j$ increases the arbitrageur's profit. In other words, $V(\TX_j)$ is $\TX_j$'s contribution to arbitrage profits, which, from another perspective, is the loss incurred by the owner of $\TX_j$.

\parhead{Summary}
The analysis above gives a clean interpretation for the maximal MEV value $\phi(s_0) + \sum_{j\in[m]} V(\TX_j)$: arbitrageur's profit is the sum of LPs' loss and the individual user losses. This interpretation helps elucidate the relationship between the maximal MEV, LVR, and user losses within a block and paves the way to develop an MEV-redistribution mechanism.

\section{Strawman Design}\label{sec:strawman}
Now, we are ready to proceed with the study of MEV-redistribution mechanisms. In the mechanism design problem, we have $n$ arbitrageurs (rather than a single arbitrageur like in \cref{sec:optimalMEV}); everyone has their own belief $v_i^*$ of the external price. The inputs of the mechanism contain (1) pool's initial state $s_0 = (x_0, y_0)$; (2) pending transactions $M$; and (3) arbitrageurs' report $\mathbf{q} = (q_1, \cdots, q_n)$.

Note that arbitrageurs may misreport beliefs and/or submit Sybil transactions. So, the report $q_i$ may differ from the true belief $v_i^*$ for any arbitrageur $i\in [n]$. For clarity, we use $\hat{s_i} = (\hat{x}_i, \hat{y}_i) \in F$ to denote the pool state corresponding to $i$'s report $q_i$, and use $s_i^* = (x_i^*, y_i^*) \in F$ to denote the pool state corresponding to $i$'s true belief $v_i^*$.

Given the characterization of the optimal MEV and its interpretation as the loss of LPs and users in previous sections, it is natural to have the following design.

\subsection{Mechanism Description}
The main idea of the strawman mechanism is to simulate each arbitrageur $i$'s maximal MEV based on their report $q_i$, sell all MEV opportunities to the one who obtains the highest MEV, and refund his/her payments to LPs and users following the interpretation in \Cref{eq:MEVnLVR}. In detail, the strawman mechanism contains the following three rules.

\begin{itemize}
    \item \textbf{Bundle generation rule:} For each arbitrageur $i\in [n]$, calculate the maximal MEV value that $i$ can obtain based on his/her report $q_i$:
    \begin{equation}
        \MEV_i \coloneqq \phi_i(s_0) + \sum_{\TX_j \in M} V_i(\TX_j),
        \label{eq:strawman_MEV}
    \end{equation}
    where $\phi_i(s_0) = \phi(s_0, q_i)$ represents the potential value of the initial state $s_0$ to arbitrageur $i$, and $V_i(\TX_j) = \max \big \{ 0, \Delta x_j \cdot q_i + \Delta y_j \big \}$ represents the actual value of $\TX_j$ to arbitrageur $i$. Note that these are measured by their report $q_i$ (rather than their true value $v_i^*$).
    
    Let's denote this winner by $w \in [n]$ who corresponds to the highest MEV value, i.e., $\MEV_w = \max_{i\in [n]} \MEV_i$ (breaking uniformly at random in the case of a tie). Then we construct a single bundle for the winner $w$ by Algorithm~\ref{alg:optimalMEV}, which takes the initial state $s_0$, all pending transactions $M$, and arbitrageur $w$'s report $q_w$ as inputs and outputs a bundle including a subset of pending transactions and some newly added MEV transactions from arbitrageur $w$.
    
    \item \textbf{Payment rule:} Arbitrageur $w$ pays the second highest MEV value in units of the num\'{e}raire. Namely, the winner $w$ pays $p_w = \max_{i\neq w} \MEV_i$; for any other arbitrageur $i\neq w$, $p_i = 0$.
    
    \item \textbf{Refund rule:} Each pending transaction $\TX_j \in M$ gets refunds
    $r_j = \frac{V_w(\TX_j)}{\MEV_w} \cdot p_w$. The remaining part, which is $\frac{\phi_w(s_0)}{\MEV_w}\cdot p_w$, is refunded to LPs in the form of swap fees.
\end{itemize}

\subsection{Analysis of Strawman Mechanism}
\begin{theorem}\label{thm:strawmantruthful}
    The strawman mechanism is truthful.
\end{theorem}

We postpone the proof to \cref{subsec:proofofstrawmantruthful}, which is conceptually similar to the reasoning behind the truthfulness of second-price auctions. Note that there is a key distinction: In the second-price auction (and most classic auction settings), a bidder's true value is fixed and will not be affected by their bid. In contrast, in the strawman mechanism above, an arbitrageur's report (analogous to the bid in the auction) determines the MEV bundle and, consequently, their profit from it (analogous to the true value). In this way, the mechanism can be seen as an auction, where misreporting not only alters one's bid (and potentially the winner), but also his/her true value if the player wins. This makes the game strategically more complex and interesting than in the second-price auction.

However, this mechanism is not Sybil-proof. As shown in \Cref{ex:Sybil} below, an arbitrageur can steal refunds by creating Sybil transactions, which decreases users' utility.

\begin{example}\label{ex:Sybil}
    Consider a similar setting to~\Cref{ex:optimalMEV}, where the trading curve is $F(x,y) = xy = 400$, with the same initial state $s_0 = (4,100)$ and swap fee $f=0$. There are three pending transactions, identical to those described in~\Cref{ex:optimalMEV}: $\TX_1=(\calX \to \calY, \delta_{\calX}^{in}=8, \delta_{\calY}^{out}=25)$, $\TX_2=(\calX \to \calY, \delta_{\calX}^{in}=30, \delta_{\calY}^{out}=12)$, and $\TX_3=(\calY \to \calX, \delta_{\calY}^{in}=20, \delta_{\calX}^{out}=10)$. None of these are Sybil transactions. Unlike~\Cref{ex:optimalMEV} where there is only one arbitrageur with a price belief of $4$, this scenario involves two arbitrageurs, each holding a different belief about the external price of $\calX$, valued at $v_1^* = 4$ and $v_2^* = 1$, corresponding to the pool states $(10, 40)$ and $(20, 20)$, respectively. 
    
    Both players truthfully report their beliefs. By definitions, the following results are derived:

\begin{table}[!ht]
\centering
\begin{tabular}{r|c|c|c|c|c|c} \hline
Arbitrageur $i$ & $q_i$ & $\phi_i(s_0)$ & \makecell[c]{$V_i(\TX_1)$} & $V_i(\TX_2)$ & $V_i(\TX_3)$ & $\MEV_i$  \\ \hline
1 & 4 & 36 & 7 & 108 & 0 & 151  \\ \hline
2 &	1 &	64 & 0 & 18 & 10 & 92\\ \hline
\end{tabular}
\end{table}
\noindent Following the strawman mechanism, arbitrageur 1 is the winner, and the mechanism constructs a bundle as illustrated in~\Cref{fig:exp1-bundle}, where all MEV transactions are arbitrageur 1's. Additionally, arbitrageur 1 is required to pay $92$ as costs, which are refunded to users and liquidity providers. Based on \Cref{def:attacker_utility}, arbitrageur 1's utility is $151 - 92 = 59$. By \Cref{def:userutility}, the utilities for users of $\TX_1$, $\TX_2$, and $\TX_3$ are $25+\frac{7}{151}\cdot 92$, $12+\frac{108}{151}\cdot 92$, and $20+\frac{0}{151}\cdot 92$, respectively.

However, arbitrageur 1 can improve his/her utility by submitting a Sybil transaction $\TX_4=(\calX \to \calY, \delta_{\calX}^{in}=260, \delta_{\calY}^{out}=271)$. This leads to the following outcome:

\begin{table}[!ht]
\centering
\begin{tabular}{r|c|c|c|c|c|c|c } \hline
Arbitrageur $i$ & $q_i$ & $\phi_i(s_0)$ & $V_i(\TX_1)$ & $V_i(\TX_2)$ & $V_i(\TX_3)$ & \textbf{$V_i(\TX_4)$} & $\MEV_i$  \\ \hline
1 & 4 & 36 & 7 & 108 & 0 & \textbf{769} & \textbf{920}  \\ \hline
2 &	1 &	64 & 0 & 18 & 10 & \textbf{0} & 92\\ \hline
\end{tabular}
\end{table}

\noindent By doing so, arbitrageur 1 still wins, but his/her utility increases to $-769 + \frac{769}{920}\cdot 92 + 920 - 92 = 59 \mathbf{+ \frac{769}{920}\cdot 92}$, while the utilities for $\TX_1$ and $\TX_2$ decrease to $25+\frac{7}{\mathbf{920}}\cdot 92$ and $12+\frac{108}{\mathbf{920}}\cdot 92$, respectively. 
\end{example}

\section{Our Mechanism}\label{sec:ourmechanism}
This section introduces the MEV-redistribution mechanism in \AMMName. Recall that the strawman mechanism sells all MEV opportunities to a single arbitrageur. In contrast, the core idea of our mechanism is to allocate the MEV opportunities to multiple arbitrageurs simultaneously.

\subsection{Mechanism Description}

\begin{itemize}

    \item \textbf{Bundle generation rule:} Our mechanism constructs the bundle following Algorithm~\ref{alg:ourMechanism}, which consists of two parts. The first part (see line \ref{algline:MCM_enumerate} - \ref{algline:MCM_enumerateEnd}) goes over pending user transactions, and each iteration starts with computing the limit state $(x_j^\ell, y_j^\ell)$ of transaction $\TX_j \in M$ and the corresponding impact on the pool $(\Delta x_j, \Delta y_j)$ by \Cref{eq:impact}. Then, the mechanism computes the potential value of $\TX_j$ to each arbitrageur $i\in [n]$ and selects the winner $w_j$ of this iteration who values $\TX_j$ the most. If this highest value $\Delta \phi_{w_j} < 0$, the mechanism skips this transaction. Otherwise, it assigns $\TX_j$ to the winner $w_j$ by constructing a ``sandwich'' on behalf of arbitrageur $w_j$. Specifically, the mechanism inserts a frontrunning transaction from state $(x_0, y_0)$ to $(x_j^l, y_j^l)$ so that $\TX_j$ executes at its limit state, followed by a backrunning transaction from the post-execution state $(x_j^l + \Delta x_j, y_j^l + \Delta y_j)$ to the initial state $(x_0, y_0)$.

    After enumerating all pending transactions, the second part of Algorithm~\ref{alg:ourMechanism} (see line \ref{algline:MCM_arbitrageStart} - \ref{algline:MCM_arbitrageEnd}) sells the rebalancing arbitrage opportunity by computing the potential value of the initial state to each arbitrageur $i\in [n]$ and assigning it to the arbitrageur who values it the most. Specifically, the mechanism adds an arbitrage transaction on behalf of the winner to reach the no-arbitrage state corresponding to his/her report.

    \item \textbf{Payment rule:} Through the bundle generation, there are $|M|+1$ winners (note that an arbitrageur may win multiple times), corresponding to $|M|$ pending transactions and the initial state. The payment rule requires each winner to pay the second highest value in units of the num\'{e}raire. Specifically, for each transaction $\TX_j$, the winner $w_j = \arg \max_{i\in [n]} V_i(\TX_j)$ pays $\max_{i\neq w_j} V_i(\TX_j)$, which is non-negative by definition, and all others pay 0. For the initial state, the winner $w' = \arg \max_{i\in [n]} \phi_i(s_0)$ pays $\max_{i\neq w'} \phi_i(s_0)$, while others pay 0.

    \item \textbf{Refund rule:} The above payment is refunded to users and liquidity providers, respectively. Specifically, the owner of transaction $\TX_j$ gets $\max_{i\neq w_j} V_i(\TX_j)$ and LPs get $\max_{i\neq w'} \phi_i(s_0)$.

\end{itemize}

We emphasize that all the quantities $\phi_i(\cdot), \Delta\phi_i, V_i(\cdot)$ are based on arbitrageur $i$'s report $q_i$. Payments and refunds are done by the mechanism.

Here, we reuse the basic setting from~\Cref{ex:Sybil} to showcase how the above mechanism operates.

\begin{example}\label{ex:ourmechanism}
    Recall that the setting is as follows: The trading curve is defined by $F(x,y) = xy = 400$, with the initial state $s_0 = (4,100)$ and swap fee $f=0$. There are three pending transactions: $\TX_1=(\calX \to \calY, \delta_{\calX}^{in}=8, \delta_{\calY}^{out}=25)$, $\TX_2=(\calX \to \calY, \delta_{\calX}^{in}=30, \delta_{\calY}^{out}=12)$, and $\TX_3=(\calY \to \calX, \delta_{\calY}^{in}=20, \delta_{\calX}^{out}=10)$. Two arbitrageurs hold a different belief about the external price of $\calX$, with values of $v_1^* = 4$ and $v_2^* = 1$, corresponding to the pool states $(10, 40)$ and $(20, 20)$, respectively. Both players report their beliefs truthfully, leading to the following outcomes:
    \begin{table}[!ht]
    \centering
    \begin{tabular}{r|c|c|c|c|c} \hline
    Arbitrageur $i$ & $q_i$ & $\phi_i(s_0)$ & \makecell[c]{$V_i(\TX_1)$} & $V_i(\TX_2)$ & $V_i(\TX_3)$  \\ \hline
    1 & 4 & 36 & \textbf{7} & \textbf{108} & 0  \\ \hline
    2 &	1 &	\textbf{64} & 0 & 18 & \textbf{10} \\ \hline
    \end{tabular}
    \end{table}
    
    According to our mechanism, arbitrageur 1 wins the MEV opportunity from $\TX_1$ and $\TX_2$, while arbitrageur 2 wins the MEV opportunity from $\TX_3$ and the initial state $s_0$. To be specific, the bundle generation rule forms a bundle as shown in~\Cref{fig:ourmechanism}. Additionally, arbitrageur 1 pays $18$, which is finally refunded to the owner of $\TX_2$; arbitrageur 2 pays $36$, which is refunded to LPs.
\end{example}

\begin{figure}[!t]
    \centering
    \includegraphics[width=0.95\linewidth]{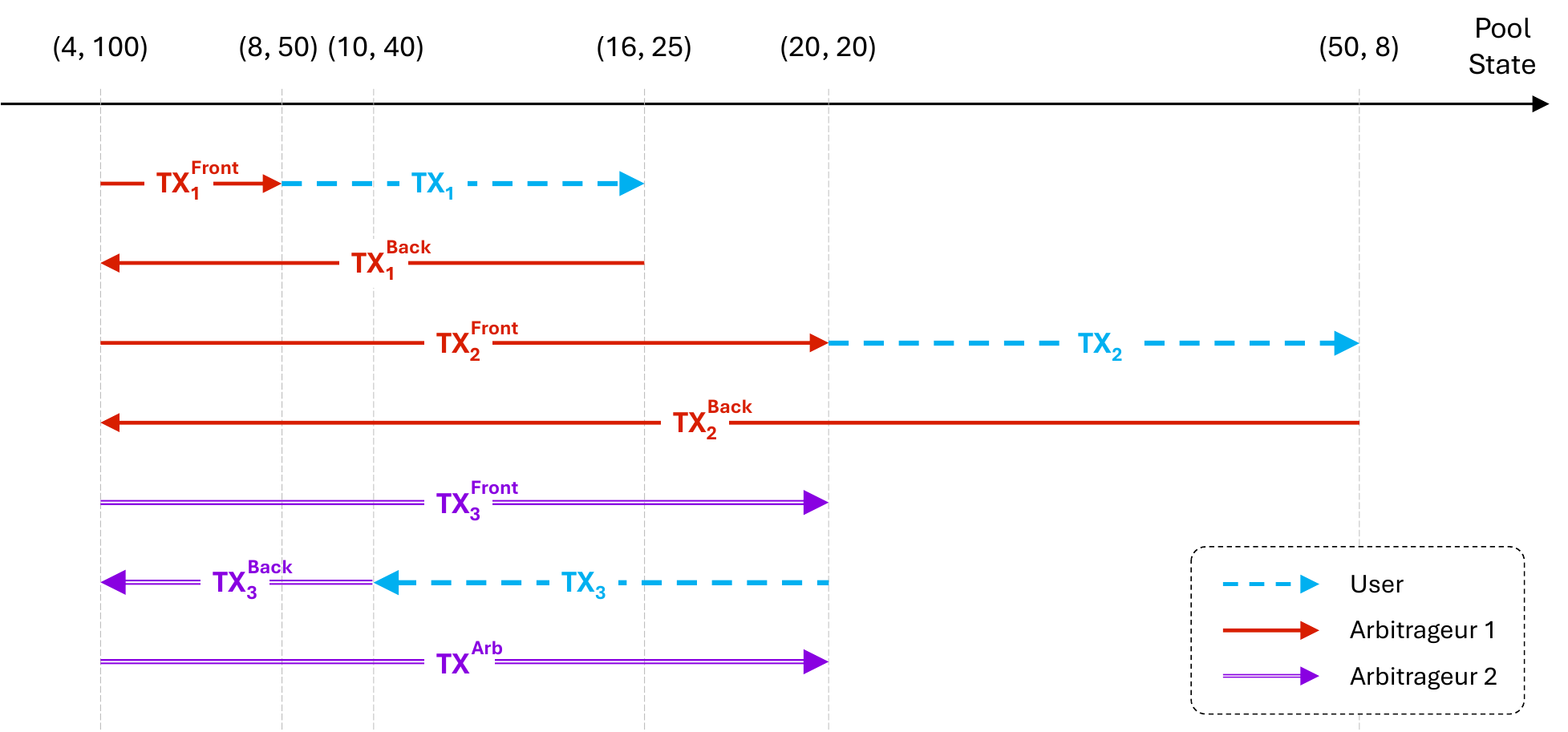}
    \caption{The bundle constructed by our mechanism for~\Cref{ex:ourmechanism}.}
    \label{fig:ourmechanism}
\end{figure}

Intuitively, each pending transaction $\TX_j \in M$ \textit{and} the initial state $s_0$ can create some MEV. 
The proposed mechanism auctions off these MEV opportunities separately, based on the value of $\TX_j$ or $s_0$ to each arbitrageur $i\in [n]$, namely, $V_i(\TX_j)$ or $\phi_i(s_0)$. The key observations are two-fold.

First, $V_i(\TX_j)$ and $\phi_i(s_0)$ solely depend on the objective information of the transaction/state (and the arbitrageur $i$'s report $q_i$), independent of other pending transactions in the pool (and other players' reports). This makes it possible to switch from a single-item auction with very complicated valuation functions, as seen in the strawman mechanism, to $|M|+1$ separate auctions.

Second, winners of these separate auctions can independently capture their MEV value, i.e., $V_{w_j}(\TX_j)$ or $\phi_{w'}(s_0)$. Enabling one arbitrageur to extract their MEV value is relatively straightforward; however, due to the ``ripple effect,''\footnote{In the CFMM context, transactions are executed sequentially. The execution of a transaction in the bundle impacts not only its own outcome (and its owner's utility), but also alters the state of the pool, which then affects the outcomes of subsequent transactions and the utility of other participants. This creates a ``ripple effect,'' where a change in one transaction can cascade through the pool and impact all others, which complicates the task of managing multiple arbitrageurs' MEV within a single bundle.} allowing all auction winners to obtain their value within a single bundle is more complex. Our mechanism overcomes this by forming the final bundle as $|M|+1$ independent sub-bundles. Each pending transaction in $M$ forms its own sub-bundle, which is a ``sandwich'' if the transaction's potential value is non-negative for at least one arbitrageur (otherwise, the sub-bundle is empty). The final sub-bundle corresponds to the initial state and consists of a single rebalancing arbitrage transaction. By having each ``sandwich'' return to the initial state, these sub-bundles are independent of one another, meaning their construction, the winner's revenue, and the associated payments do not interfere with or rely on other sub-bundles. %

\begin{algorithm}[htbp]
	\caption{Bundle Generation Rule of MEV-Redistribution Mechanism in \AMMName}
	\label{alg:ourMechanism}
	\KwIn{An initial state $s_0=(x_0,y_0)$, a set of pending transactions $M$, and arbitrageurs' reports $\mathbf{q} = (q_1, \cdots, q_n)$, each corresponding to a state $\hat{s_i} = (\hat{x}_i, \hat{y}_i)$.
	}
	\KwOut{A bundle to be executed by CFMM}
	
	\BlankLine
	
	\BlankLine

    \For{each $j \in [1:|M|]$}{ \label{algline:MCM_enumerate}

    \If{$\TX_j = (\calX \to \calY, \delta_{\calX}^{in}, \delta_{\calY}^{out})$}{\label{line:MCM_preprocessBegin}
        $(x_j^\ell, y_j^\ell)$ is the limit state satisfying $y_j^\ell - F_{y}(x_j^\ell + \delta_{\calX}^{in}) = \delta_{\calY}^{out}$.
        
        Let $\Delta x \gets \delta_{\calX}^{in}$, $\Delta y \gets -\delta_{\calY}^{out}$.

    }
    \ElseIf{$\TX_j=(\calY \to \calX, \delta_{\calY}^{in}, \delta_{\calX}^{out})$}{
        $(x_j^\ell, y_j^\ell)$ is the limit state satisfying $x_j^\ell - F_x(y_j^\ell + \delta_{\calY}^{in}) = \delta_{\calX}^{out}$.
    
        Let $\Delta x \gets -\delta_{\calX}^{out}$, $\Delta y \gets \delta_{\calY}^{in}$. 
    } 

    \For{each $i \in [n]$}{
        Let $\Delta \phi_i \gets \Delta x \cdot q_i + \Delta y$. 
    }
    Let $\Delta\phi_{w_j} \gets \max_{i\in[n]} \Delta \phi_i.$  {\hfill\tcp{Arbitrageur $w_j$ values $\TX_j$ the most}}
    
    \If{$\Delta\phi_{w_j} \geq 0$}{   \label{line:MCM_insertBegin}
        \If{$x_0 < x_j^\ell$}{
            Insert a transaction on behalf of arbitrageur $w_j$ 
            $\TX =\big( \calX \to \calY, \delta_{\calX}^{in}=x_j^\ell-x_0, \delta_{\calY}^{out}=y_0-y_j^\ell \big)$.
        }
        \ElseIf{$x_0 > x_j^\ell$}{
            Insert a transaction on behalf of arbitrageur $w_j$
            $\TX =\big(\calY \to \calX, \delta_{\calY}^{in}=y_j^\ell-y_0, \delta_{\calX}^{out}=x_0 - x_j^\ell \big)$.
        }
    
        Place the user transaction $\TX_j$. 

        Let the current state $(x', y') \gets (x_j^\ell + \Delta x, y_j^\ell + \Delta y)$.

        \If{$x' < x_0$}{
            Insert a transaction on behalf of arbitrageur $w_j$ 
            $\TX =\big( \calX \to \calY, \delta_{\calX}^{in}=x_0-x', \delta_{\calY}^{out}=y'-y_0 \big)$.
        }
        \ElseIf{$x' > x_0$}{
            Insert a transaction on behalf of arbitrageur $w_j$
            $\TX =\big(\calY \to \calX, \delta_{\calY}^{in}=y_0 - y', \delta_{\calX}^{out}=x' - x_0 \big)$.
        }
    }   \label{line:MCM_insertEnd}
    
    }\label{algline:MCM_enumerateEnd}
    \For{each $i\in [n]$}{\label{algline:MCM_arbitrageStart}
        Let $\phi_i \gets (x_0 \cdot q_i + y_0) - (\hat{x}_i \cdot q_i + \hat{y}_i)$.
    }

    Let $\phi_{w'} \gets \max_{i\in[n]}  \phi_i.$  {\hfill\tcp{Arbitrageur $w'$ values \textsf{LVR} the most}}
    
    \If{$x_0 < \hat{x}_{w'}$}{
        Add a transaction on behalf of arbitrageur $w'$ 
        $\TX = \big( \calX \to \calY, \delta_{\calX}^{in}=\hat{x}_{w'}-x_0, \delta_{\calY}^{out}=y_0-\hat{y}_{w'} \big)$.
    }
    \ElseIf{$x_0 > \hat{x}_{w'}$}{
        Add a transaction on behalf of arbitrageur $w'$ 
        $\TX =\big(\calY \to \calX, \delta_{\calY}^{in}=\hat{y}_{w'} - y_0,  \delta_{\calX}^{out}=x_0 - \hat{x}_{w'} \big)$.
    }\label{algline:MCM_arbitrageEnd}
    
\end{algorithm}

\subsection{Analysis of Our Mechanism}

\begin{theorem}\label{theorem:truthful}
    Our mechanism is truthful.
\end{theorem}

\begin{theorem}\label{thm:sybilproof}
   Our mechanism is Sybil-proof.
\end{theorem}

\Cref{thm:sybilproof} tells us that, given the report (which may differ from the true belief), Sybil transactions do not reduce any user's utility. This implies that a user's utility depends on the arbitrageurs' reports while being independent of their decisions to use Sybil transactions or not. Additionally, \Cref{theorem:truthful} demonstrates that for any arbitrageur $i \in [n]$, if $i$ submits no Sybil transaction (i.e., $S_i = \emptyset$), truthful reporting is a dominant strategy. These two properties make our mechanism seem very close to the ideal mechanism. However, there is a subtle and tricky issue: an arbitrageur's utility is jointly determined by their report and their Sybil behavior. Submitting Sybil transactions to maximize profits may introduce the incentive for arbitrageurs to misreport their beliefs, which in turn may affect users' utilities. 

Our ultimate goal is to incentivize arbitrageurs to truthfully report their beliefs, under which users get their deserved utilities in our mechanism. 
This can be achieved by showing that there is a strategy profile $(v^*, S)$ that is a Nash equilibrium. 
In the following, we show that this is indeed true when arbitrageurs have some Sybil budget $({b}_i^{\calX}, {b}_i^{\calY})$ and each arbitrageur's belief is drawn from some known distribution $\calD_i$. We use $\calD$ to denote the collection of all $\calD_i$'s.

In particular, we show the following theorem:

\begin{theorem}\label{thm:NE}
    There is a Sybil strategy $S_i(v_i^*,{b}_i^{\calX}, {b}_i^{\calY}, \calD)$ such that under our mechanism, using $(v_i^*, S_i)$ is a Nash equilibrium for every arbitrageur $i\in [n]$. 
\end{theorem}

We postpone the proofs of the above theorems  to~\cref{subsec:ourtruthfulprf},~\ref{subsec:oursybilproof}, and~\ref{subsec:NEproof}.

\parhead{Remark.} The bundle generation rule in our mechanism creates the maximal MEV over all possible bundles, which is $\max_{i\in [n]} \phi_i(s_0) + \sum_{j\in [|M|]} \max_{i\in [n]} V_i(\TX_j)$.

\section{Evaluation}
\label{sec:evaluation}

In this section, we evaluate the efficiency of \AMMName---specifically, whether it improves execution results (i.e., price) for users and reduces loss for liquidity providers. For users, we compare \AMMName to two notable application-level solutions, \uniswapx and \cowswap, to evaluate if orders filled through \AMMName achieve better execution prices than in \uniswapx or \cowswap. For LPs, we compare the LVR of Uniswap v2 LPs with and without the MEV-redistribution mechanism.

\subsection{Evaluation Methodology}
\parhead{Assumptions} We assume that arbitrageurs have sufficient capital to execute trades between CEXs and DEXs, making it potentially profitable for them to use their own assets to fulfill users' orders. 
Additionally, the liquidity on CEXs is ample enough to ensure that these arbitrage activities do not materially affect the off-chain price. We justify these assumptions in two ways: First, many arbitrageurs engage in CEX-DEX arbitrage in practice~\cite{heimbach2024non}; second, the liquidity on CEXs is generally substantial, as indicated by their high trading volumes (e.g., the 24-hour trading volume of ETH on Binance is \$8 billion as of October 5th, 2024~\cite{binance2024ethereumprice}).
Besides, for simplicity, we assume that the gas usage for a transaction of an order in \uniswapx or \cowswap is the same when the order is filled by \AMMName. We also assume that the priority fee paid by arbitrageurs in \AMMName is up to 1 Gwei. 
We note that this is a reasonable assumption because the arbitrageurs in \AMMName compete for the opportunity to fulfill orders at the CFMM side instead of participating in auctions at the block producer side.

\parhead{Distribution of arbitrageurs' beliefs}
To simulate arbitrageurs' beliefs in external
prices, we obtain historical token price data from Binance~\cite{binance2024data} over a one-year period, from September 1st, 2023, to August 31st, 2024. Since not all tokens are traded on CEXs, our evaluation focuses on the eight popular tokens: BTC, ETH, USDC, USDT, DAI, LINK, MATIC, and PEPE. Specifically, we obtain candlestick data from Binance with one-second intervals for these tokens.
The arbitrageurs' beliefs about a token's price in a block are simulated by a distribution between the highest and lowest prices of the token during the time slot of the block (the specific type of distribution and the number of arbitrageurs are discussed in each experiment).

\parhead{Historical trades}
We collect all orders on \uniswapx and \cowswap within the same time range. Specifically, we gather orders from the DEX Analytics Platform~\cite{dextrades2024dataset} and cross-check them against Dune records~\cite{dune2024}. Our evaluation focuses on orders where either the buy or sell side is a stablecoin (USDC, USDT, or DAI), while the counterpart is a token for which we have price data from Binance.
In the end, we collect 344,936 orders in \uniswapx and 100,618 orders in \cowswap for comparison.
To evaluate the efficiency in reducing LVR, we also collect the state of two Uniswap v2 pools (WETH-USDC and WETH-USDT) in each block over the same time range.

\subsection{Results}\label{subsec:evaluationresults}
\parhead{Better execution prices} We replay historical trades to compute the execution price on \AMMName and compare it with the actual execution price on \uniswapx or \cowswap. 
In the simulation, we vary the number of participating arbitrageurs in \AMMName and explore various belief distributions that arbitrageurs might have for a token. We model searchers' beliefs using a Gaussian distribution~\cite{feller1991introduction} with a centered mean and controlled standard deviation, as well as a Pareto distribution~\cite{arnold2014pareto} with a shape parameter $\alpha = 1.5$, where most arbitrageurs expect lower token valuations, but a few anticipate significantly higher ones.

Additionally, since liquidity pools can charge different swap fees, we tested swap fee settings ranging from 0 to 0.5\% (we excluded 1\% because the token pairs we focus on are typically concentrated in pools with lower fees~\cite{etherscan2024dex}). Note that the information on swap fees is not available from historical trades on \uniswapx and \cowswap, as those orders may not involve a pool (e.g., some are filled using solvers' assets; see \cref{sec:empiricalevidence}).  

\begin{figure}[htbp!]
    \centering
    \begin{subfigure}[b]{0.47\textwidth}
        \centering
        \includegraphics[width=\textwidth]{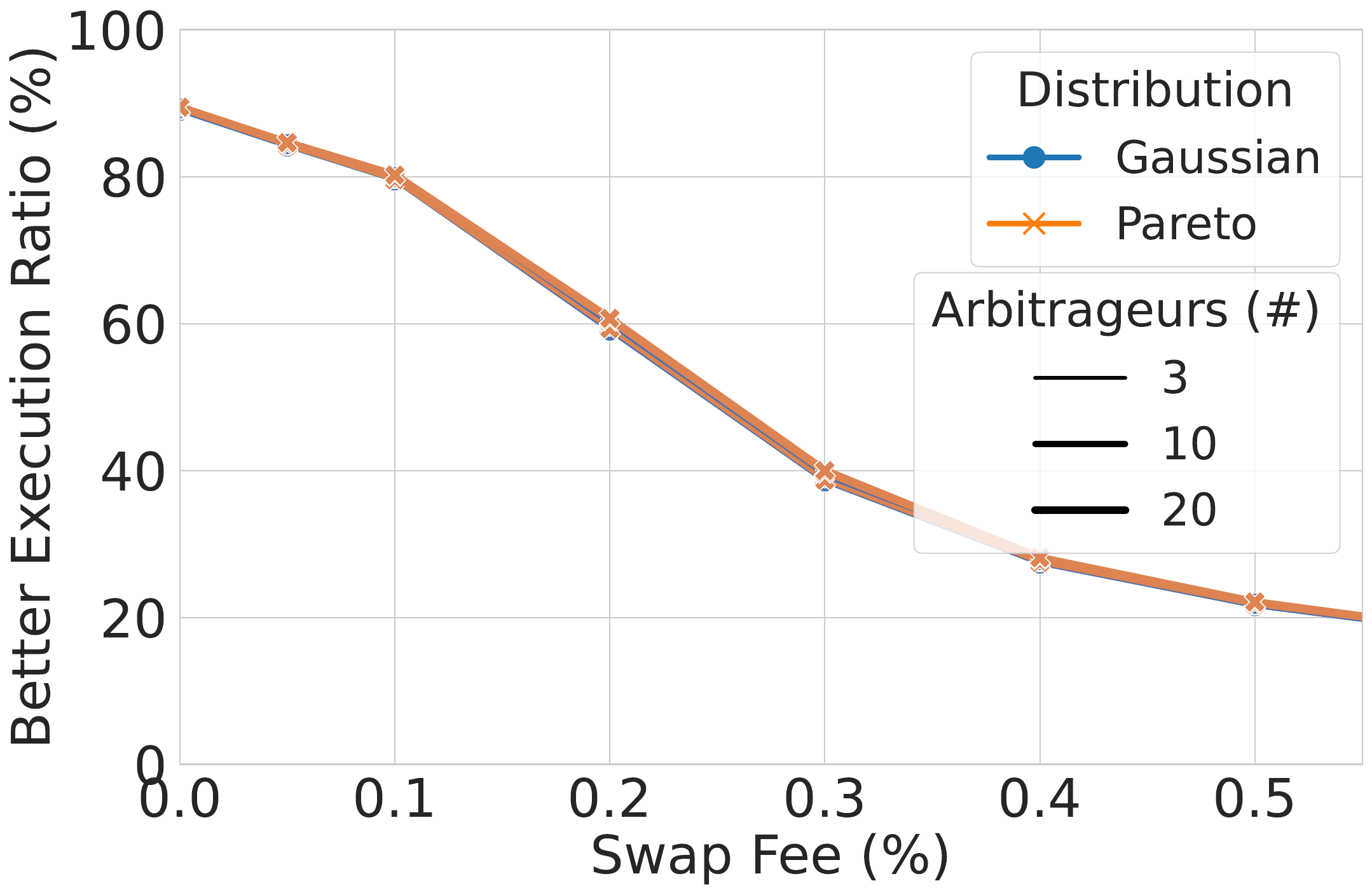}
        \caption{\uniswapx}
        \label{fig:uniswapx-ratio}
    \end{subfigure}
    \hfill
    \begin{subfigure}[b]{0.47\textwidth}
        \centering
        \includegraphics[width=\textwidth]{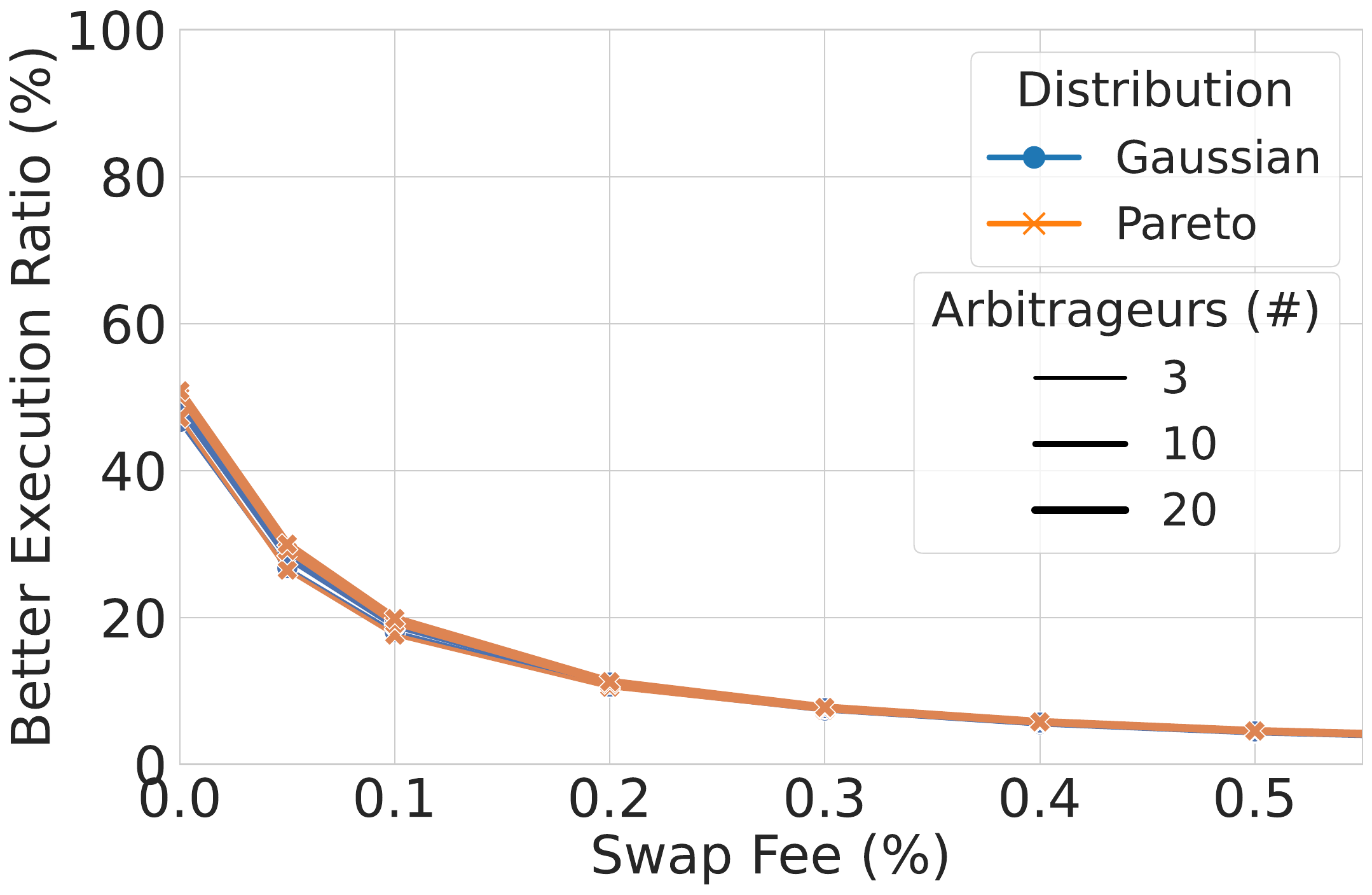}
        \caption{\cowswap}
        \label{fig:cowswap-ratio}
    \end{subfigure}
    \caption{Percentage of orders with execution prices better than those on \uniswapx or \cowswap. Each marker represents the overall result under different settings of arbitrageurs and swap fees.}
    \label{fig:better-execution-price}
\end{figure}

As shown in~\Cref{fig:better-execution-price}, we find that execution prices in \AMMName are better than those in \uniswapx for 89\% of orders when the swap fee is 0. The percentage of orders with better execution prices decreases as the swap fee increases, but it remains 40\% when the swap fee is 0.3\%. Compared to \cowswap, our mechanism can provide nearly equally competitive execution prices when the swap fee is 0. The distribution of prices and the number of arbitrageurs have no obvious impact on the results. This may be due to the narrow price range on Binance, which affects visible differences.

\begin{figure}[htbp!]
    \centering
    \begin{subfigure}[b]{0.49\textwidth}
        \centering
        \includegraphics[width=\textwidth]{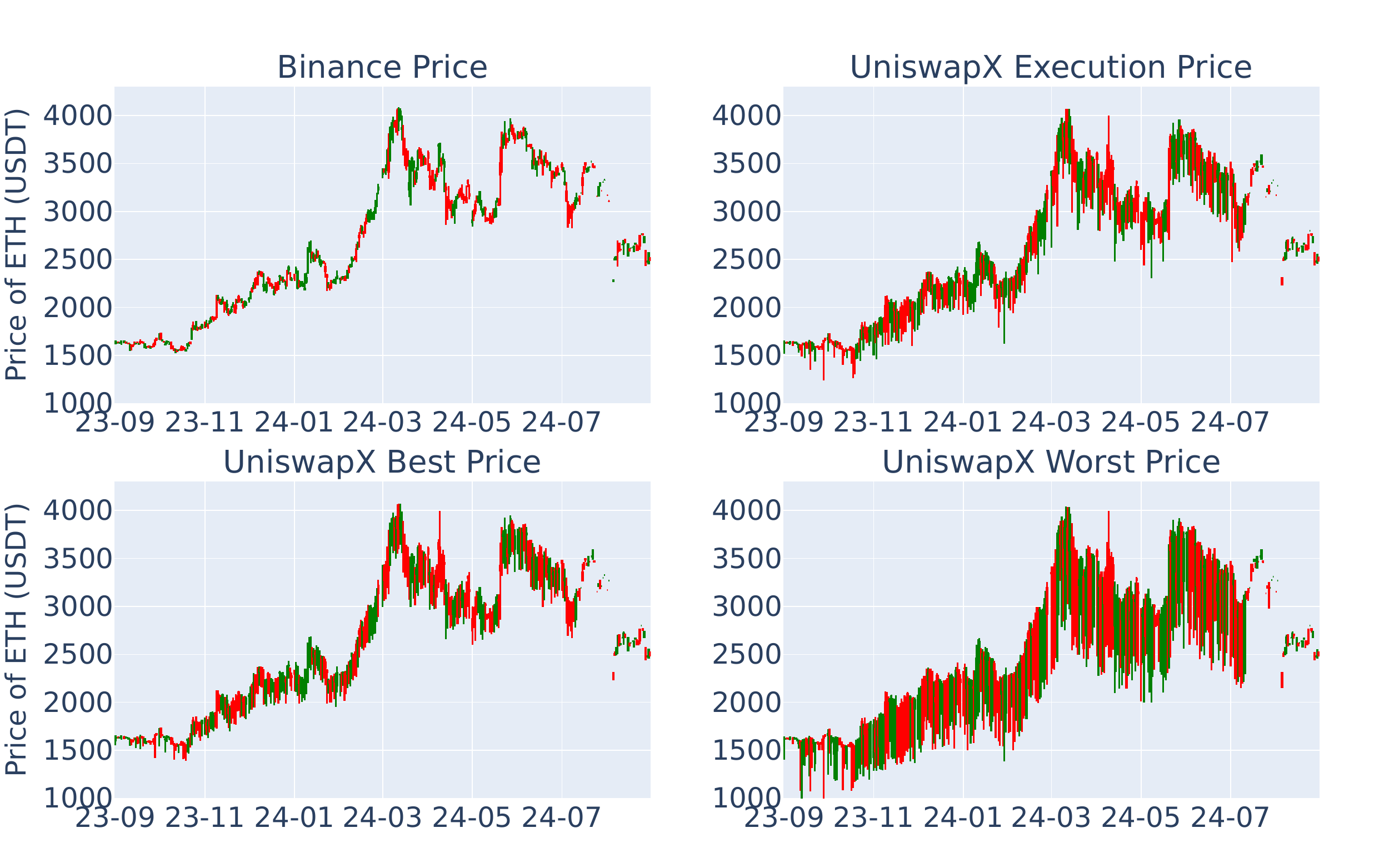}
        \caption{Swap ETH for USDT}
        \label{fig:price-eth-usdt}
    \end{subfigure}
    \hfill
    \begin{subfigure}[b]{0.49\textwidth}
        \centering
        \includegraphics[width=\textwidth]{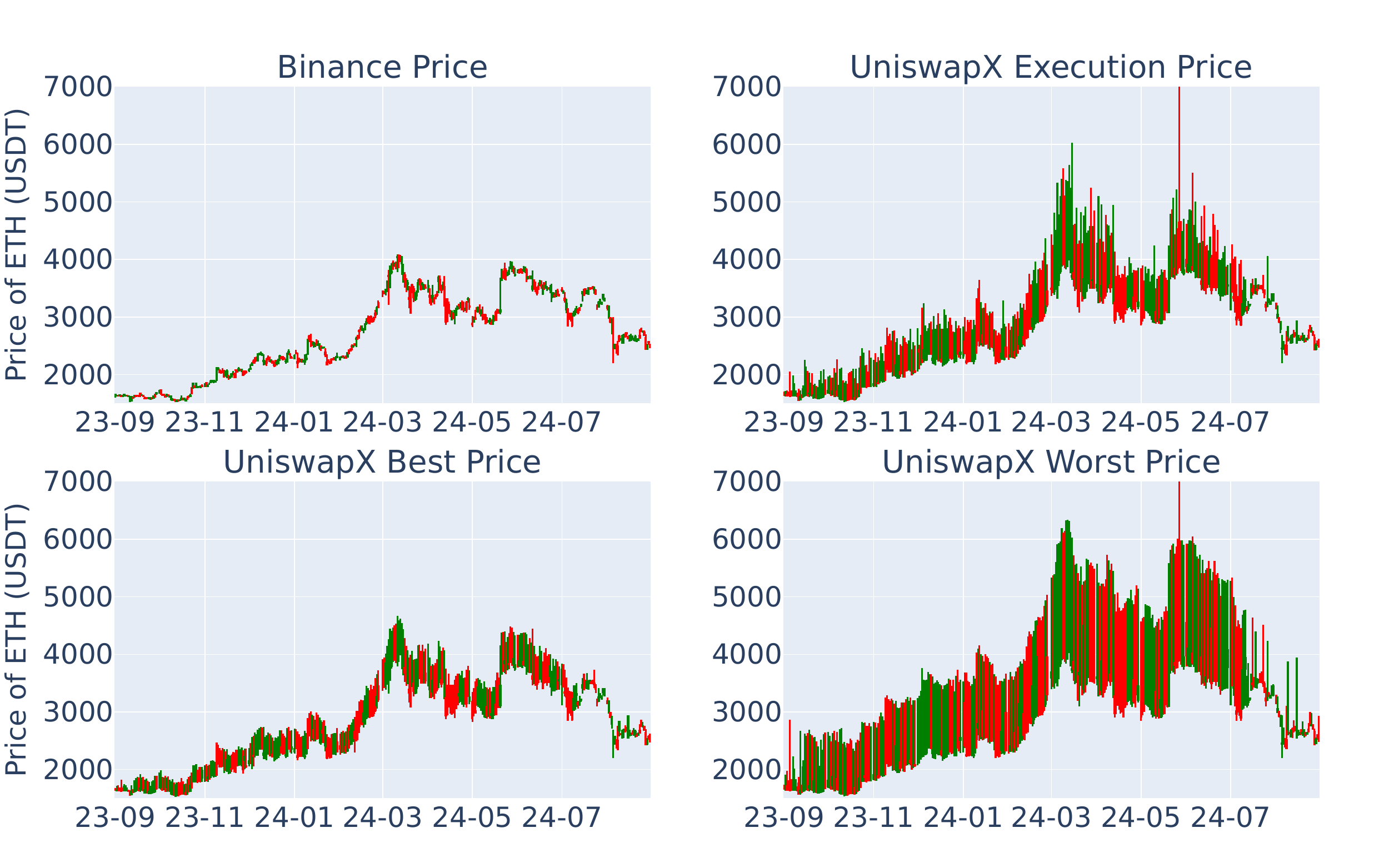}
        \caption{Swap USDT for ETH}
        \label{fig:price-usdt-eth}
    \end{subfigure}
    \caption{Candlestick charts of Binance and \uniswapx ETH/USDT prices over time. Missing data indicates that our dataset does not contain relevant orders on \uniswapx for those days.}
    \label{fig:price-binance-uniswapx}
\end{figure}

To understand why \AMMName outperforms \uniswapx in most orders, we select ETH/USDT---the most frequently traded token pair in our dataset---to compare price differences on Binance and \uniswapx over time. 
For \uniswapx, we analyze the execution price, the best price (the starting price of \uniswapx's Dutch auction), and the worst price (the ending price in the auction, representing the lowest price users are willing to accept to fulfill orders). As shown in~\Cref{fig:price-binance-uniswapx}, we observe that prices on Binance are better than the best price on \uniswapx in most cases: for swapping ETH for USDT, the price of ETH on Binance is higher, and for swapping USDT for ETH, the price of ETH on Binance is lower. Given that arbitrageurs' beliefs follow a distribution around the prices on Binance in our simulation, \AMMName can provide users with better execution prices by leveraging the more favorable prices on CEX.

\begin{figure}[thbp!]
    \centering
    \begin{subfigure}[b]{0.47\textwidth}
        \centering
        \includegraphics[width=\textwidth]{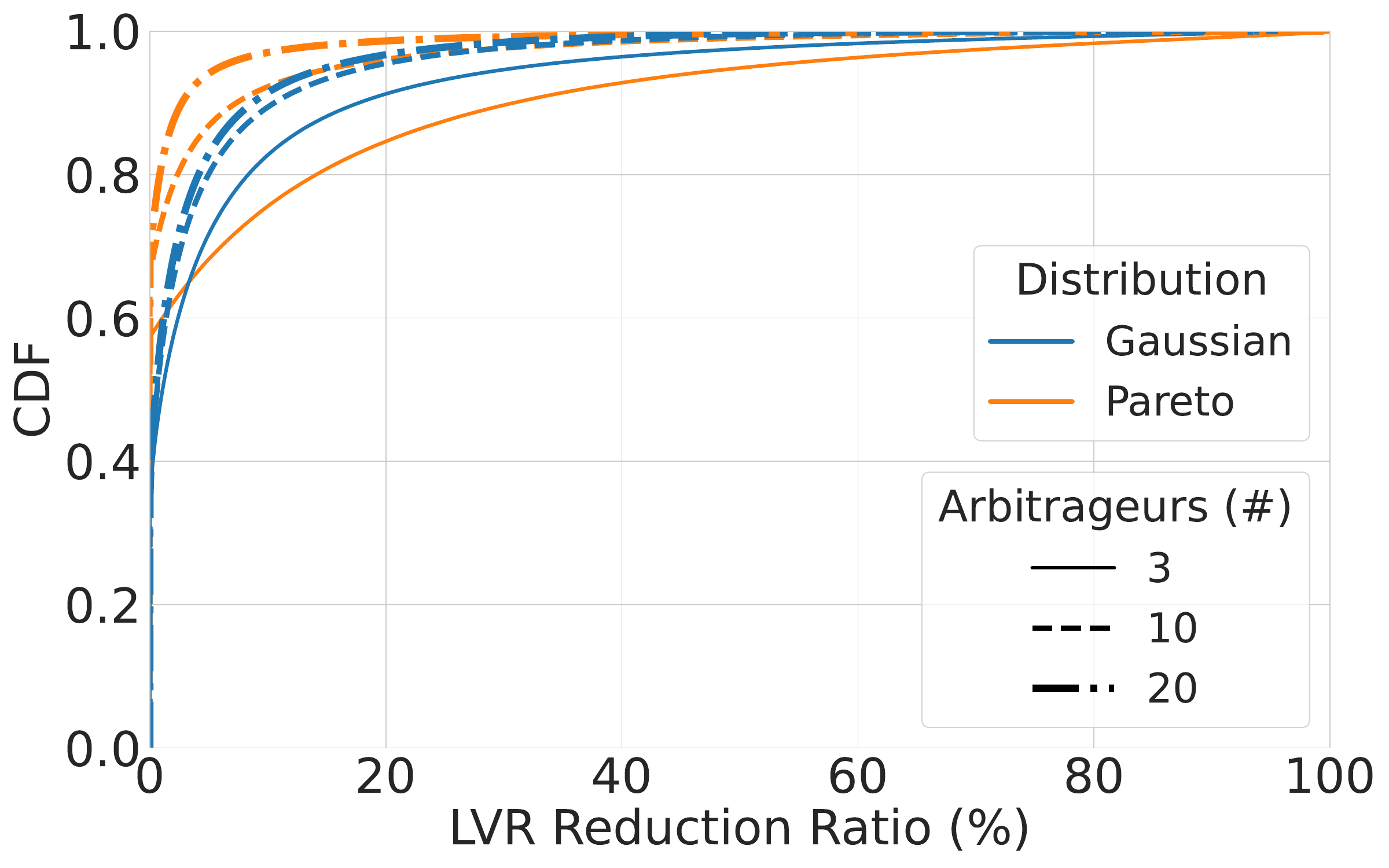}
        \caption{WETH-USDC (Uniswap v2)}
        \label{fig:weth-usdc}
    \end{subfigure}
    \hfill
    \begin{subfigure}[b]{0.47\textwidth}
        \centering
        \includegraphics[width=\textwidth]{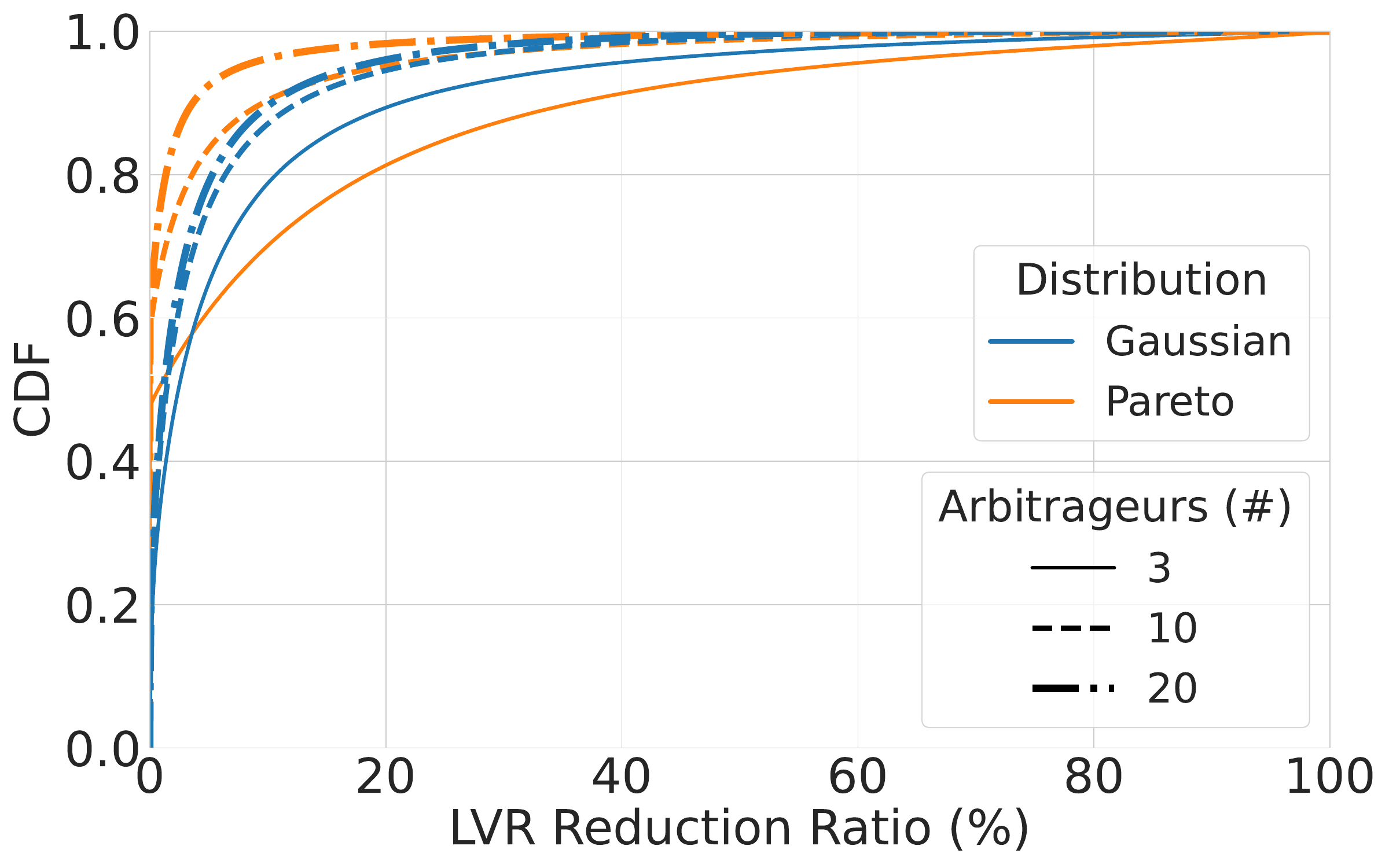}
        \caption{WETH-USDT (Uniswap v2)}
        \label{fig:weth-usdt}
    \end{subfigure}
    \caption{The CDF of LVR reduction for WETH-USDC and WETH-USDT using \AMMName for different numbers of arbitrageurs and price distributions. A smaller LVR reduction ratio on the x-axis indicates greater reduction efficiency.}
    \label{fig:better-lvr}
\end{figure}

\parhead{Reducing LVR} 
For each block in our evaluation period, we computed the LVR incurred by two Uniswap v2 liquidity pools (WETH-USDC and WETH-USDT) when the arbitrageurs perform CEX-DEX arbitrages, with and without using \AMMName. By dividing the loss under \AMMName by the LVR without \AMMName, we can evaluate the proportion of LVR reduction achieved by \AMMName. Similar to the previous evaluation, we also test different settings for the number of arbitrageurs and the distribution of their price beliefs for a token.

As shown in~\Cref{fig:better-lvr}, we can see that \AMMName effectively reduces LVR for both liquidity pools, and the reduction improves as more arbitrageurs participate. The heightened competition drives arbitrageur bids closer, ultimately minimizing LVR.
For example, in over half of the blocks, the LPs in the WETH-USDC pool will only suffer at most 0.5\% of the LVR they would experience without using \AMMName, when there are 20 arbitrageurs with beliefs following a Gaussian distribution.

\section{Conclusion, Discussion, and Future Works}
We have introduced $\AMMName$, a CFMM aided by an MEV-redistribution mechanism that can capture MEV at the application level and then refund it to the relevant participants (e.g., users and liquidity providers in this paper). We proved that the MEV-redistribution mechanism is truthful and Sybil-proof. Moreover, $\AMMName$ is designed to not rely on arbitrageurs' sophistication. Below, we discuss implementation considerations and future directions.

\subsection{Implementation Considerations}
This paper focuses on the mechanism design problem, and we leave its concrete implementation for future work, for which we do not foresee significant challenges. 
\AMMName consists of a CFMM pool (a smart contract) and an MEV-redistribution mechanism.
The main implementation task is to protect the confidentiality of arbitrageurs' reports\footnote{For instance, if the reports $\mathbf{q}$ of all arbitrageurs are public, the user of an $\calX \to \calY$ transaction can pretend to be an arbitrageur and submit a fake report $\bar{q}-\epsilon$ where $\bar{q} = \max_{i\in [n]} q_i$, in order to get more refunds. Protecting the confidentiality of arbitrageurs' reports can prevent such manipulation.}. Since most smart contract platforms do not offer confidentiality, the mechanism needs to run in an off-chain environment. This also allows off-loading computation burdens to save on gas costs.

Several tools are available to implement the mechanism. The most straightforward option is to use hardware-based Trusted Execution Environments (TEEs), such as Intel SGX~\cite{costan2016intel} and TDX~\cite{intel2023tdx}, which are readily available from cloud computing platforms and achieve near-native performance. %
The high-level workflow is to run the mechanism in a TEE, which accepts (encrypted) inputs from users and arbitrageurs, computes the bundles, and invokes the CFMM pool's smart contract at the appropriate time. Appendix~\ref{app:implementation} presents more details.

An alternative implementation is to use a combination of fully homomorphic encryption and zero-knowledge proofs; because sensitive information only needs to remain private for a short period (e.g., minutes), a lower security level may be used to improve performance.

\subsection{Future Directions}
\parhead{Solver Behaviors} The empirical studies presented in \cref{sec:empiricalevidence} reveal a discrepancy between the intended design of protocols and the actual behaviors of solvers in both \uniswapx and \cowswap. To summarize, in UniswapX, where solvers are expected to utilize diverse liquidity sources to provide better solutions, over 84\% of filled orders turned out to be filled with solvers' own assets. Similarly, \cowswap is known for matching orders with complementary trading intents (hence the name ``coincidences of wants''), but 76.94\% of batches consist of just a single order, indicating that matching instances within batches are very rare. 

Although both protocols are solver-based, a key difference lies in the nature of solver competition. In \cowswap, solvers compete internally within the protocol, with only the solver providing the best solution being rewarded and their solution selected for settlement on-chain. This incentivizes solvers to align more closely with user interests, as they must offer the best solution to win the competition. In contrast, in UniswapX, any solver who can successfully fulfill an order can do so by simply sending a transaction. As a result, solvers are incentivized to fulfill orders at the specified limit price without optimizing for the best execution. This difference in competitive structure creates a divergence in solver incentives between the two protocols, with \cowswap's model better aligning solver behavior with user outcomes, as supported by our evaluation results in~\cref{subsec:evaluationresults}. 

A potential avenue for future research would be to investigate solvers' performance in solver-based protocols using a broader set of metrics and analyze the underlying factors driving these behaviors. This could provide valuable insights into improving protocol design, particularly to align solver incentives more effectively with user needs.

\parhead{MEV Capturing in other DeFi Applications} Results in this paper can naturally extend to the CFMM pool where both $\mathcal{X}$ and $\mathcal{Y}$ are risky assets or stablecoins. This extension involves two key modifications. First, the arbitrageur's beliefs are about the external prices of two assets, denoted as $v_{\mathcal{X}}^*$ and $v_{\mathcal{Y}}^*$. Accordingly, the arbitrageurs need to report both values in the mechanism. Second, the potential value of a pool state in~\Cref{eq:potential} is adjusted to be $(x \cdot v_{\mathcal{X}}^* + y \cdot v_{\mathcal{Y}}^*) - (x^* \cdot v_{\mathcal{X}}^* + y^* \cdot v_{\mathcal{Y}}^*)$, where $(x^*, y^*)$ represents the no-arbitrage state at which $\left|\frac{\partial F/\partial x}{\partial F/\partial y}\right| = v_{\mathcal{X}}^*/v_{\mathcal{Y}}^*$. Consequently, all subsequent related definitions are updated to reflect this modification.

Beyond the non-atomic arbitrage considered in this work, other types of MEV, such as atomic arbitrage within or across DEXs, are also worth studying. 
Furthermore, similar methods to capturing and redistributing MEV at the application level could potentially be applied to other DeFi applications, such as lending platforms and oracles (c.f. oracle extractable values~\cite{oev}), where MEV might manifest differently.

\section*{Acknowledgements}
We thank DEX Analytics and Allium for providing the historical order data.


\bibliographystyle{IEEEtran}

\newpage
\appendix
\section{Analysis of the orders in \uniswapx and \cowswap}
\label{sec:empiricalevidence}

In this section, we analyze how the users' orders are fulfilled on \uniswapx and \cowswap. In particular, we are interested in the following questions:
\begin{enumerate}
    \item How many orders are filled through direct exchanges between users and solvers?
    \item How many orders are filled through batch auctions?
\end{enumerate}

The first question examines the percentage of orders directly filled by solvers' assets rather than liquidity pools in DEXs. The second question evaluates the effectiveness of the batch auction process, where orders can be fulfilled together by a solver as a group, known as a {\em batch}. If a batch contains only one order, it suggests that the batch auction's effectiveness is lower than expected, as there is no counterpart order within the batch.

The answers to these questions reveal how solvers fulfill orders---whether they simply use their own assets or aggregate orders in a more complex manner.

\parhead{Orders collection} We collect 663,831 orders on \cowswap from the DEX Analytics Platform~\cite{dextrades2024dataset} during the period from September 1, 2023, to August 31, 2024. We cross-check these \cowswap orders with records on Dune~\cite{cowprotocol2024cowswap} and find that the number of orders matched. 
Similarly, we collect 726,789 orders on \uniswapx from the DEX Analytics Platform within the same date range and cross-checked them against Dune records. The results from both sources matched, indicating the accuracy of our collected data.

\parhead{The number of orders filled through direct exchanges}
An order can be filled through various solutions, such as existing liquidity pools like Uniswap v2/v3, or through direct exchange between the solver and the user. Since there is no public dataset that directly indicates how an order in \uniswapx or \cowswap is filled, to answer the first question, we need to differentiate the solutions by which an order is filled.
A key observation in answering this question is that if an order is directly filled through token exchanges between a filler and a user, {\em the process does not involve any liquidity pool}.

Based on this observation, we use a heuristic to determine whether an order is directly filled through a filler-user token exchange. If a transaction includes only a user order, and this order is filled without involving any liquidity pool, we infer that it must be filled by direct exchange.

We apply this heuristic to the historical orders we collected, and the results are shown in~\Cref{fig:searcher-orders}.
A surprising finding is that over 84\% of the orders on \uniswapx were filled through direct token exchange between solvers and users, suggesting that this solution is widely used by active solvers on \uniswapx. In comparison, a smaller percentage of orders on \cowswap were filled through direct exchange---the percentage ranges from 12.5\% to 17.7\%.

One possible explanation for the difference between these two protocols is that different sets of solvers are active in \uniswapx and \cowswap, respectively. For example, an active solver\footnote{0xfbeedcfe378866dab6abbafd8b2986f5c1768737} on \uniswapx did not participate in \cowswap.

\begin{figure}[thbp!]
    \centering
    \begin{subfigure}[b]{0.49\textwidth}
        \centering
        \includegraphics[width=\textwidth]{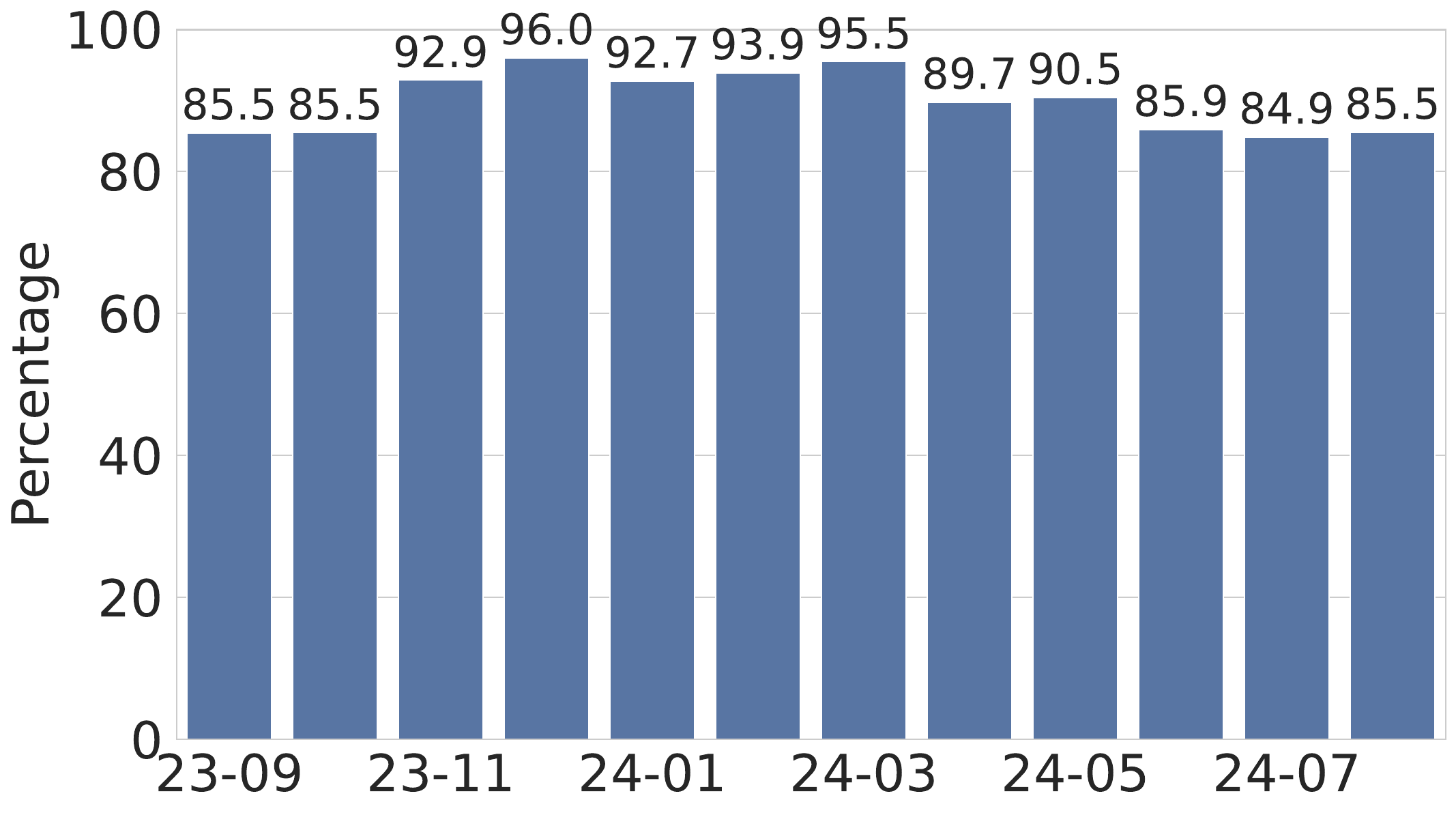}
        \caption{\uniswapx}
        \label{fig:uniswapx-searcher-orders}
    \end{subfigure}
    \hfill
    \begin{subfigure}[b]{0.49\textwidth}
        \centering
        \includegraphics[width=\textwidth]{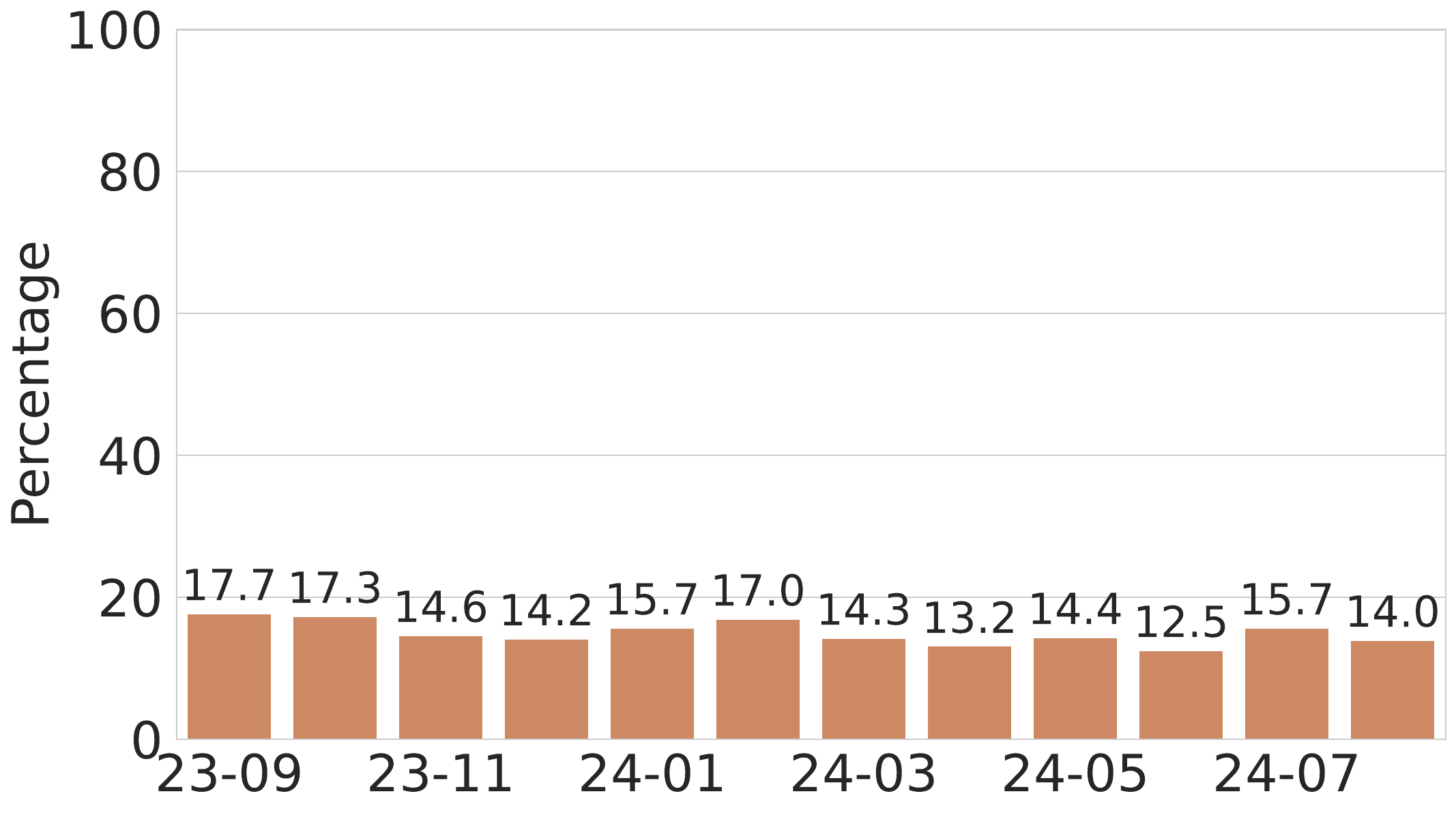}
        \caption{\cowswap}
        \label{fig:cowswap-searcher-orders}
    \end{subfigure}
    \caption{Percentage of orders filled by direct exchange between users and solvers.}
    \label{fig:searcher-orders}
\end{figure}

\parhead{The number of orders per batch}
When collecting historical orders on \uniswapx and \cowswap, we also recorded the hash of the transaction in which each order was filled and the solver who fulfilled it. Note that a batch refers to a group of orders filled by a solver within a single transaction. Using the collected data, we can compute the number of orders filled within each batch. The distribution of the number of orders per batch is shown in~\Cref{fig:order-per-tx}.

An interesting observation from~\Cref{fig:uniswapx-order-per-tx} is that over 99\% of batches on \uniswapx contained only a single filled order. In contrast, as shown in~\Cref{fig:cowswap-order-per-tx}, 17.6\% of batches on \cowswap contained two filled orders, while fewer than 6\% of batches contained more than three orders, and fewer than 0.5\% included more than five orders.

\begin{figure}[htbp!]
    \centering
    \begin{subfigure}[b]{0.49\textwidth}
        \centering
        \includegraphics[width=\textwidth]{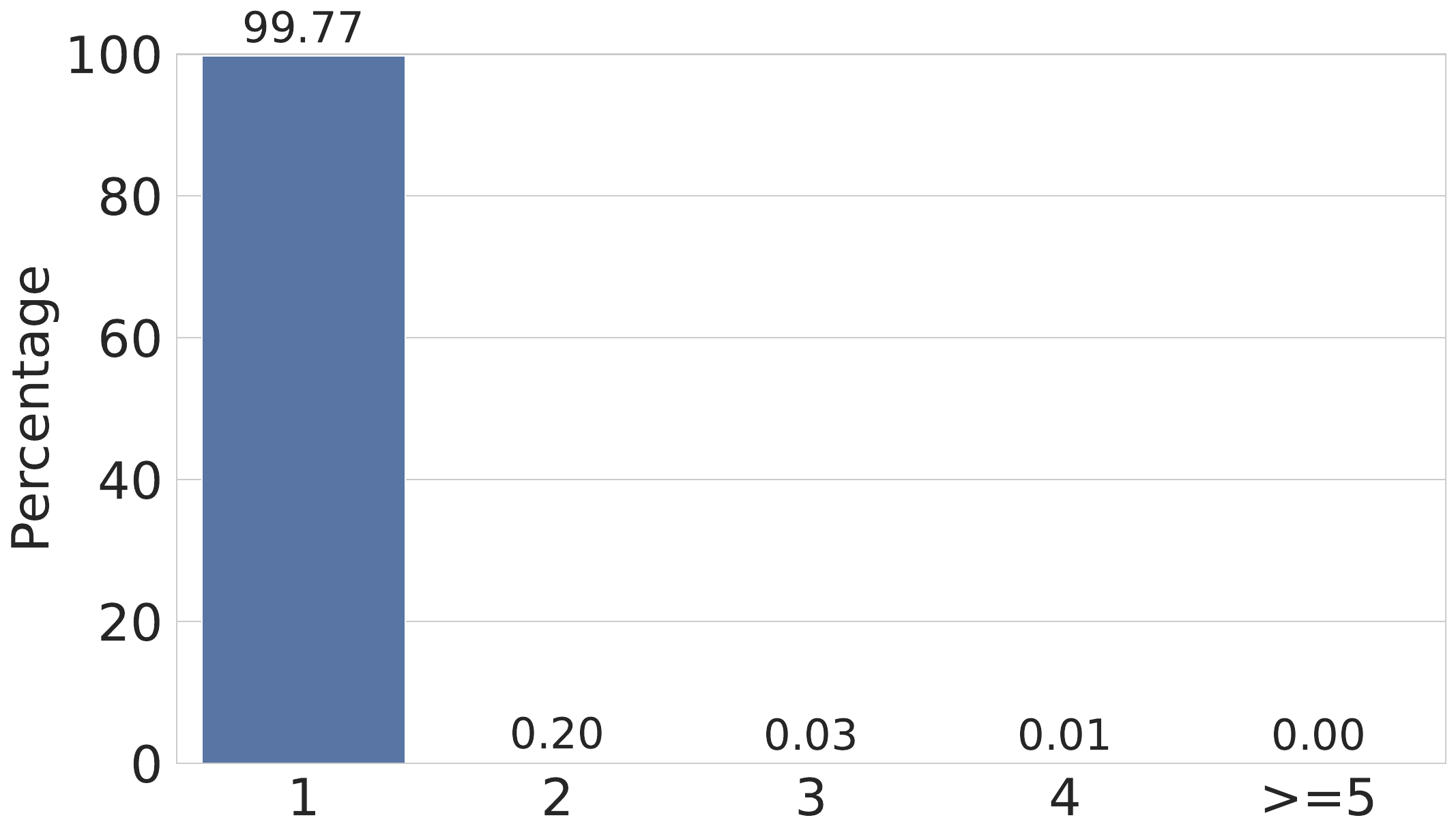}
        \caption{\uniswapx}
        \label{fig:uniswapx-order-per-tx}
    \end{subfigure}
    \hfill
    \begin{subfigure}[b]{0.49\textwidth}
        \centering
        \includegraphics[width=\textwidth]{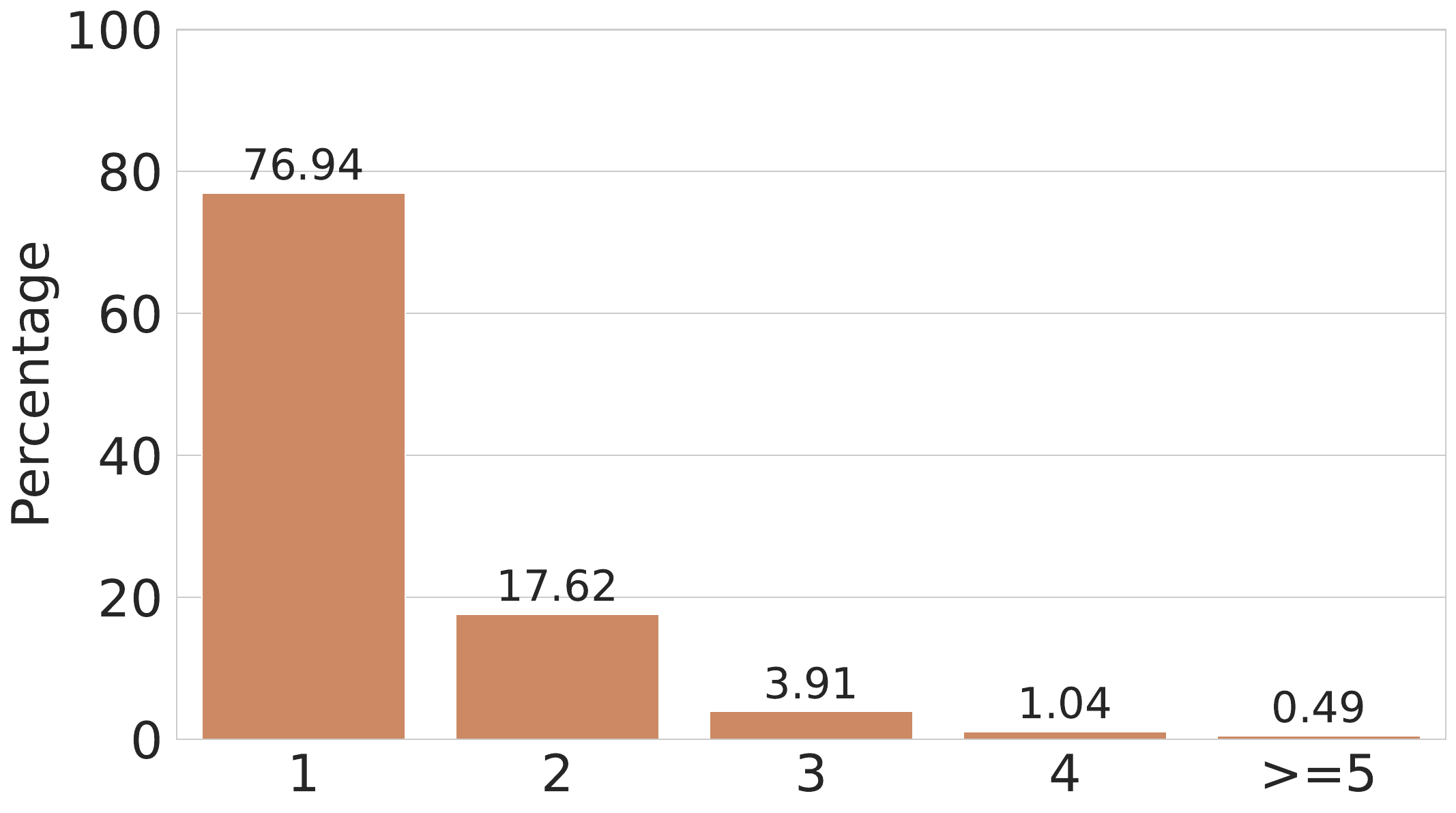}
        \caption{\cowswap}
        \label{fig:cowswap-order-per-tx}
    \end{subfigure}
    \caption{Distribution of the number of orders per batch.}
    \label{fig:order-per-tx}
\end{figure}

\section{MEV Optimization Problem under Public Orders}\label{sec:optimal}
Without loss of generality, we discuss the MEV optimization problem from the perspective of an arbitrary arbitrageur. The setting is summarized as follows:

\begin{itemize}
    \item There is a CFMM pool with a risky asset $\mathcal{X}$ and a num\'{e}raire asset $\mathcal{Y}$;
    \item The pool's initial state is $s_0 = (x_0, y_0)$;
    \item There is a set of public pending transactions $M = \{\TX_1, \cdots, \TX_m\}$ from users;
    \item The arbitrageur's belief in the external price of the risky asset $\calX$ is $v^*$. Note that for an arbitrary $v^* >0$, there is exactly one pool state $s^* = (x^*, y^*)$ at which the pool's marginal exchange rate $\left|\frac{\partial F/\partial x}{\partial F/\partial y}\right| = v^*$. So sometimes we use $v^*$ and $(x^*, y^*)$ interchangeably to represent the arbitrageur's belief.
\end{itemize}
Note that in this setting, arbitrageurs do not need to create Sybil transactions in advance, as they can insert arbitrary transactions when constructing the bundle. The MEV optimization problem can be formalized by defining the arbitrageur's strategy space and utility function as follows.

\begin{definition}[Strategy Space in the MEV Optimization Problem]
    Given an initial state $s_0=(x_0,y_0)$ and a set of user transactions $\set{\TX_j}_{j\in[m]}$, an arbitrageur can construct a bundle by selecting a subset of users' transactions $T\subseteq [m]$, creating a set of his/her own transactions $A = \set{\TX_j}_{j\in[m+1:m+a]}$, and picking an execution order (a permutation) over these transactions $\pi:[|T|+|A|]\to T\cup A$.
\end{definition}

\begin{definition}[Utility Function in the MEV Optimization Problem]
    Arbitrageur's profit is 
    \begin{equation}\label{eq:utilityPublicOF}
        U\big(T, A, \pi\big) = 
        \sum_{j\in[|T|+|A|], \pi(j)\in A}
        \Bigg[ \left(x_{j-1}-x_j\right) \cdot v^* + (y_{j-1}-y_j) \Bigg],
    \end{equation}
    where $v^*$ is the arbitrageur's belief in the external price of $\calX$, and $(x_j, y_j)$ is the pool's state after the $j$-th transaction in the bundle.
\end{definition}

This utility function is a simplified version of Formula (\ref{eq:attacker_utility}), as it only contains the profits from arbitrage transactions represented by Formula (\ref{eq:arbtxs}). The arbitrageur's objective is to find a strategy $\big(T, A, \pi \big)$ that maximizes their MEV profits, which we refer to as an \emph{optimal MEV strategy}.

\subsection{Optimal MEV Strategy}
This section elaborates on the strategy stated in Algorithm~\ref{alg:optimalMEV} for the arbitrageur to efficiently extract the maximal MEV value from any set of user transactions. %

Roughly speaking, Algorithm~\ref{alg:optimalMEV} can be divided into two parts. The first part (see line~\ref{algline:enumerate} -~\ref{algline:enumerateEnd}) enumerates all transactions and tries to ensure that if a transaction is profitable, it will be executed exactly at its limit state (defined below). Specifically, for each transaction $\TX_j$, we first preprocess it in line~\ref{line:preprocessBegin} -~\ref{line:preprocessEnd}: Recall that each transaction specifies the maximum input amount ($\delta^{in}$) and minimum output amount ($\delta^{out}$). Based on the transaction information, $\TX_j$'s \textit{limit state} $s_j^\ell = (x_j^\ell, y_j^\ell)$ is defined as the state at which, when executed, the transaction will pay exactly the maximum amount of input token and receive the minimum amount of output token. For example, suppose $\TX_j =(\calX \to \calY, \delta_{\calX}^{in}, \delta_{\calY}^{out})$, then the limit state $(x_j^\ell, y_j^\ell)$ is the state satisfying $y_j^\ell - F_{y}(x_j^\ell + \delta_{\calX}^{in}) = \delta_{\calY}^{out}$. Likewise, if $\TX_j=(\calY \to \calX, \delta_{\calY}^{in}, \delta_{\calX}^{out})$, then $(x_j^\ell, y_j^\ell)$ satisfies $x_j^\ell - F_x(y_j^\ell + \delta_{\calY}^{in}) = \delta_{\calX}^{out}$. Furthermore, we denote the transaction $\TX_j$'s impact on the trading pool when executed at its limit state by $(\Delta x_j, \Delta y_j)$ which is defined as
\begin{equation*}  
    (\Delta x_j, \Delta y_j) \coloneqq
    \left\{  
        \begin{array}{lr}  
            (\delta_{\calX}^{in}, -\delta_{\calY}^{out}), & \TX_j = (\calX \to \calY, \delta_{\calX}^{in}, \delta_{\calY}^{out}); \\  
            (-\delta_{\calX}^{out}, \delta_{\calY}^{in}), & \TX_j=(\calY \to \calX, \delta_{\calY}^{in}, \delta_{\calX}^{out}).    
        \end{array}  
    \right.
\end{equation*}
meaning that the transaction will bring $\delta^{in}$ amount of input tokens to the pool and take $\delta^{out}$ amount of output tokens away from the pool.

Next, the algorithm decides whether to execute $\TX_j$ depending on the transaction's potential value $\Delta \phi_j$, which is defined as
\begin{equation*}
    \Delta \phi_j \coloneqq \Delta x_j \cdot v^* + \Delta y_j,
\end{equation*}
where $v^*$ is the arbitrageur's belief in the external price. If $\Delta \phi_j < 0$, the algorithm directly ignores this transaction and moves on to the next iteration; otherwise (see line~\ref{line:insertBegin} -~\ref{line:insertEnd}), inserts an arbitrageur's transaction such that the after-execution state is exactly $(x_j^\ell, y_j^\ell)$, then executes $\TX_j$.

After that, we come to the second part of the algorithm (see line~\ref{algline:arbitrage} -~\ref{algline:arbitrageEnd}), which compares the current pool state $(x,y)$ with the no-arbitrage state $s^* = (x^*, y^*)$ corresponding to the belief $v^*$, and ensures that the pool will stop at the state $s^*$ by adding an arbitrage transaction if needed.

The high-level idea about the Algorithm~\ref{alg:optimalMEV} and an intuitive example can be found in~\cref{subsec:optimalMEV}.

\section{Missing Proofs}

\subsection{Proof of \Cref{theorem:publicOF}}\label{subsec:mevproof}
Before going into the formal proofs, we first present a lemma (\Cref{lemma:upperbound}) that characterizes the upper bound of the arbitrageur's profits. (Later, we will prove that Algorithm~\ref{alg:optimalMEV} can exactly get this value.)

\begin{lemma}\label{lemma:phigeq0}
    For any belief $v^*$ and any pool state $s=(x,y)$, the state's potential value for the arbitrageur $\phi(s, v^*) \geq 0$.
\end{lemma}

\begin{proof}[Proof of \Cref{lemma:phigeq0}]
    Fix an arbitrary belief $v^*$, there is a unique pool state $(x^*, y^*) \in F$ at which $\frac{y^*}{x^*} = v^*$. For any state $s = (x,y)\in F$, there are three cases:
    
    \textbf{Case 1:} $x = x^*$ and $y = y^*$. $\phi(s) = 0$ by definition.

    \textbf{Case 2:} $x > x^*$ and $y < y^*$. $v^* = \frac{y^*}{x^*} > \frac{y^* - y}{x - x^*} \Rightarrow \phi(s) > 0$. 

    \textbf{Case 3:} $x < x^*$ and $y > y^*$. $v^* = \frac{y^*}{x^*} < \frac{y - y^*}{x^* - x} \Rightarrow \phi(s) > 0$.
\end{proof}

\begin{lemma}\label{lemma:upperbound}
    Given an initial state $s_0=(x_0,y_0)$, a set of users transactions $\set{\TX_j}_{j\in[m]}$, and the arbitrageur's price belief $v^*$, the arbitrageur's profit is upper bounded by $\phi(s_0, v^*) + \sum_{j\in[m]} V(\TX_j)$.
\end{lemma}

\begin{proof}[Proof of \Cref{lemma:upperbound}]
    Fix an arbitrary sequence of the mixed arbitrageur's and users' transactions $\big(\TX_{\pi(1)}, \cdots, \TX_{\pi(m+a)}\big)$, where $\TX_{\pi(j)}$ is a user transaction if $\pi(j)\in [m]$ and it is the arbitrageur's transaction otherwise. We will inductively show that after executing the first $k$ transactions, the arbitrageur's profit $U_k \leq \phi(s_0) + V^k$ where $V^k \coloneqq \sum_{j\in[k], \pi(j) \in [m]} V(\TX_{\pi(j)})$. This will imply that after executing all $m+a$ transactions, the arbitrageur's profit is upper bounded by $\phi(s_0) + \sum_{k\in[m]} V(\TX_k)$.

    Let $s_{k}$ be the state after executing $\TX_{\pi(k)}$ and $\phi_k = \phi(s_{k})$. We focus on $U_k+\phi_k$ where $U_0+\phi_0 = 0 + \phi(s_0) = \phi(s_0)$. Let $V(\TX_{\pi(k)})=0$ if it is the arbitrageur's transaction. We will show that $(U_k + \phi_k) - (U_{k-1} + \phi_{k-1}) \leq V^k - V^{k-1} = V(\TX_{\pi(k)})$ for all $k\in [m+a]$, which will imply our desired statement $U_k \leq \phi(s_0) + V^k$ for all $k\in [m+a]$, as $\phi_k \geq 0$ always holds according to \Cref{lemma:phigeq0}. For each $k\in [m+a]$, there are two cases: $\TX_{\pi(k)}$ is from a user or the arbitrageur.

    \textbf{Case 1:} $\TX_{\pi(k)}$ is a user transaction. In this case, $U_k = U_{k-1}$. So it suffices to show that $\phi_k - \phi_{k-1} \leq V(\TX_{\pi(k)})$. 
    According to \Cref{eq:potential}, $\phi_k - \phi_{k-1} = (x_k - x_{k-1})\cdot v^* - (y_{k-1} - y_k)$. 
    \begin{itemize}
        \item If $\TX_{\pi(k)} = (\calX \to \calY, \delta_{\calX}^{in}, \delta_{\calY}^{out})$, $\phi_k - \phi_{k-1} \leq \delta_{\calX}^{in} \cdot v^* - \delta_{\calY}^{out} = \Delta\phi_{\pi(k)}$.
        
        \item If $\TX_{\pi(k)}=(\calY \to \calX, \delta_{\calY}^{in}, \delta_{\calX}^{out})$, $\phi_k - \phi_{k-1} \leq - \delta_{\calX}^{out} \cdot v^* + \delta_{\calY}^{in} = \Delta\phi_{\pi(k)}$.
    \end{itemize}
    By the definition of $V(\cdot)$, the inequality $\Delta\phi_{\pi(k)} \leq V(\TX_{\pi(k)})$ naturally holds, so we have $\phi_k - \phi_{k-1} \leq V(\TX_{\pi(k)})$ which concludes the first case.
    
    \textbf{Case 2:} $\TX_{\pi(k)}$ is the arbitrageur's transaction. In this case, $V^k = V^{k-1}$. Then it suffices to show that $(U_k + \phi_k) - (U_{k-1} + \phi_{k-1}) \leq 0$. By definition (see \Cref{eq:utilityPublicOF}), we have $U_k - U_{k-1} = (x_{k-1} - x_k)\cdot v^* + (y_{k-1} - y_k)$, which is exactly $\phi_{k-1}-\phi_k$ according to \Cref{eq:potential}. Thus, $(U_k + \phi_k) - (U_{k-1} + \phi_{k-1}) = 0$, concluding the second case.

    This finishes the proof of \Cref{lemma:upperbound}.
\end{proof}

Next, we go back to \Cref{theorem:publicOF} and prove that Algorithm~\ref{alg:optimalMEV} captures this upper bound.

\begin{proof}[Proof of \Cref{theorem:publicOF}]
    Algorithm~\ref{alg:optimalMEV} enumerates all given transactions and outputs an execution sequence, including both user transactions and newly inserted/added transactions from the arbitrageur him/herself. 

    Let $s_{j}$ be the state after the $j$-th iteration where $j\in [m]$, $\phi_{j}$ be the value of state $s_{j}$ according to \Cref{eq:potential}, and $U_{j}$ be the utility of the arbitrageur at that state according to \Cref{eq:utilityPublicOF}. We focus on $U_j + \phi_j$.

    Initially, we have $U_0 = 0$ and $\phi_0 = \phi(s_0)$, thus $U_0 + \phi_0 = \phi(s_0)$.

    Then for each iteration $j\in[m]$, we decide whether to execute $\TX_j$ based on $\Delta\phi_j$. If $\Delta\phi_j<0$, we choose to skip, so $U + \phi$ remains unchanged, namely, $U_j + \phi_j = U_{j-1}+\phi_{j-1}$. Otherwise (i.e., $\Delta\phi_j \geq 0$), we execute an arbitrageur's transaction (if needed) followed by the user's transaction $\TX_j$. Note that the execution of the arbitrage transaction changes the pool state from $(x_{j-1}, y_{j-1})$ to $(x_j^\ell, y_j^\ell)$, during which $U_{j-1}$ and $\phi_{j-1}$ change to be $U'$ and $\phi'$ respectively. Then we have
        \begin{align}
            U' + \phi' 
            &= U_{j-1} + (x_{j-1} - x_j^\ell)\cdot v^* + (y_{j-1} - y_j^\ell) + (x_j^\ell \cdot v^* + y_j^\ell) - (x^* \cdot v^* + y^*) \nonumber  \\
            &= U_{j-1} + (x_{j-1} - x_j^\ell)\cdot v^* + (y_{j-1} - y_j^\ell) + (x_j^\ell - x_{j-1})\cdot v^* + (y_j^\ell - y_{j-1}) + \phi_{j-1}   \nonumber   \\
            &= U_{j-1} + \phi_{j-1}.
            \label{eq:U+phi}
        \end{align}
    Next, the execution of user transaction $\TX_j$ changes $U'$ and $\phi'$ into $U_j$ and $\phi_j$, respectively. Here, we have $U_j = U'$ and $\phi_j = \phi' + \Delta\phi_j$. Thus, $U_j + \phi_j = U' + \phi' + \Delta\phi_j$. In both cases ($\Delta\phi_j<0$ or $\Delta\phi_j \geq 0$), we have $U_j + \phi_j = U' + \phi' + \max\{0, \Delta\phi_j\} = U_{j-1}+\phi_{j-1} + V\big(\TX_j \big)$.
     
    Iteratively, we have $U_m + \phi_m = \phi(s_0) + \sum_{j\in [m]} V(\TX_j)$ after the $m$-th iteration. At the end of the algorithm, we conclude the strategy with an arbitrage transaction such that we will stop at the state $s^* = (x^*, y^*)$ corresponding to the arbitrageur's belief $v^*$. In this process, $U_m$ and $\phi_m$ change to be $U$ and $\phi(s^*)$, where $\phi(s^*) = 0$ by definition. According to \Cref{eq:U+phi}, $U= U_m + \phi_m = \phi(s_0) + \sum_{j\in [m]} V(\TX_j)$. 

    This finishes the proof.
\end{proof}

\subsection{Proof of \Cref{thm:strawmantruthful}}\label{subsec:proofofstrawmantruthful}

\begin{proof}[Proof of \Cref{thm:strawmantruthful}]
    Fix an arbitrary arbitrageur $i$, its true belief $v_i^*$, and the reports $\textbf{q}_{-i}$ of the other arbitrageurs. We will show that arbitrageur $i$'s utility is maximized by setting $q_i = v_i^*$.

    Let $\overline{\MEV} = \max_{k\neq i} \MEV_k$ denote the highest MEV value by some other arbitrageur $k\neq i$. Note that the definition of \textit{truthfulness} assumes arbitrageur $i$ adds no Sybil transactions, namely, $S_i = \emptyset$. This implies that by reporting $q_i$, if the induced MEV value $\MEV_i \geq \overline{\MEV}$, then $i$ wins and receives utility $\sum_{\TX_{\pi(j)} \in A_i} \Big[ \left(x_{j-1}-x_j\right) \cdot v_i^* + (y_{j-1}-y_j) \Big] - p_i$ according to (\ref{eq:attacker_utility}), where $A_i$ is the set of newly inserted arbitrage transactions for $i$ in the bundle generation step; if the report $q_i$ induces that $\MEV_i < \overline{\MEV}$, then $i$ loses and receives utility 0.

    Denote the MEV value corresponding to $i$'s true belief $v_i^*$ by $\MEV_i^*$. \Cref{theorem:publicOF} tells us that $\sum_{\TX_{\pi(j)} \in A_i} \Big[ \left(x_{j-1}-x_j\right) \cdot v_i^* + (y_{j-1}-y_j) \Big] \leq \MEV_i^*$. We conclude by considering two cases. First, if $\MEV_i^* < \overline{\MEV}$, the maximum utility that arbitrageur $i$ can obtain is $0$, and $i$ achieves this by reporting truthfully (and losing). Second, if $\MEV_i^* \geq \overline{\MEV}$, the maximum utility that arbitrageur $i$ can obtain is $\max\Big\{0, \sum_{\TX_{\pi(j)} \in A_i} \Big[ \left(x_{j-1}-x_j\right) \cdot v_i^* + (y_{j-1}-y_j) \Big] - \overline{\MEV}\Big\} = \MEV_i^* - \overline{\MEV}$, and $i$ achieves this by reporting truthfully (and winning).
\end{proof}

\subsection{Proof of \Cref{theorem:truthful}}\label{subsec:ourtruthfulprf}
\begin{proof}[Proof of \Cref{theorem:truthful}]
    Fix an arbitrary arbitrageur $i$, its true belief $v_i^*$, the reports $\textbf{q}_{-i}$ of all other arbitrageurs, and their Sybil transactions $S_{-i}$. To prove the truthfulness, we only need to focus on the case where arbitrageur $i$ submits no Sybil transaction, i.e., $S_i = \emptyset$. The bundle generation process shows that the final bundle is composed of $|M|+1$ sub-bundles, corresponding to $|M|$ pending transactions (for transactions that failed to be executed, the sub-bundle is empty) and the initial state. Denote the sub-bundle constructed for transaction $\TX_j \in M$ and initial state $s_0$ by $\SB(\TX_j)$ and $\SB(s_0)$, respectively. Then we can rewrite arbitrageur $i$'s utility defined in Equation (\ref{eq:attacker_utility}) by traversing each sub-bundle and calculating the profit from it as follows:
    \begin{equation*}
        u(\emptyset, q_i; S_{-i}, \mathbf{q}_{-i}) 
        = \sum_{e \in M\cup \{s_0\}} \textsf{monetary value of arbitrage transactions in SB($e$) - payment for it}. 
    \end{equation*}
    
    Consider the scenario where arbitrageur $i$ misreports by setting $q_i \neq v_i^*$. For each element $e \in M\cup \{s_0\}$, there are three cases to discuss regarding the winner of its sub-bundle.

    \textbf{Case 1:} Arbitrageur $i$ is the winner when reporting truthfully, but not when misreporting. In this case, arbitrageur $i$'s profit from the sub-bundle constructed when reporting truthfully is the loss of misreporting, which we will show is non-negative. The element $e$ under discussion is either a pending transaction $\TX_j$ or the initial state $s_0$. If the former, arbitrageur $i$'s profit from this sub-bundle is
    \begin{align}
        & \overbrace{(x_0 - x_j^l)\cdot v_i^* + (y_0 - y_j^l)}^{\text{Front-running: $(x_0, y_0) \to (x_j^l, y_j^l)$}}
        \ \ + \quad \overbrace{(x_j^l + \Delta x_j - x_0) \cdot v_i^* + (y_j^l + \Delta y_j - y_0)}^{\text{Back-running: $(x_j^l + \Delta x_j, y_j^l + \Delta y_j) \to (x_0, y_0)$}} \quad - \overbrace{p_i(\TX_j)}^{\text{Payment for $\TX_j$}}    \nonumber    \\
        =& \ \Delta x_j \cdot v_i^* + \Delta y_j - \max_{k\neq i} V_k(\TX_j), \label{eq:proofcase1}
    \end{align}
    which is non-negative. This is because the fact that arbitrageur $i$ is the winner when reporting truthfully implies that $\Delta x_j \cdot v_i^* + \Delta y_j \geq \max_{k\neq i} V_k(\TX_j)$.

    If the latter, the sub-bundle only contains a rebalancing arbitrage transaction from $(x_0, y_0)$ to the state $(x^*, y^*)$ corresponding to $v_i^*$ (i.e., $y^*/x^*=v_i^*$), from which arbitrageur $i$'s profit is
    \begin{align*}
        & (x_0 - x^*) \cdot v_i^* + (y_0 - y^*) - \max_{k\neq i} \big( (x_0 \cdot q_k + y_0) - (\hat{x}_k \cdot q_k + \hat{y}_k) \big)    \\
        =& \phi_i(s_0) - \max_{k\neq i} \phi_k(s_0),
    \end{align*}
    which is also non-negative as $\phi_i(s_0) \geq \max_{k\neq i} \phi_k(s_0)$.
    
    \textbf{Case 2:} Arbitrageur $i$ is not the winner when reporting truthfully, but wins when misreporting. In this case, the profit from the sub-bundle is the gain of misreporting, which we will show is negative. If the sub-bundle corresponds to a transaction $\TX_j$, the profit is still $\Delta x_j \cdot v_i^* + \Delta y_j - \max_{k\neq i} V_k(\TX_j)$, which becomes negative in this case, as arbitrageur $i$ loses when reporting truthfully. 
    
    If the sub-bundle corresponds to $s_0$, the profit becomes 
    \begin{align}
        h(\hat{x}_i, \hat{y}_i) &\coloneqq (x_0 - \hat{x}_i) \cdot v_i^* + (y_0 - \hat{y}_i) - \max_{k\neq i} \phi_k(s_0)   \nonumber    \\
        &= (x_0 \cdot v_i^* + y_0) - (\hat{x}_i \cdot v_i^* + \hat{y}_i) - \max_{k\neq i} \phi_k(s_0),  \label{eq:proofcase2}
    \end{align}
    where $(\hat{x}_i, \hat{y}_i)$ corresponding to $q_i$ (i.e., $\hat{y}_i/\hat{x}_i=q_i$). We will show that $h(\hat{x}_i, \hat{y}_i)$ is upper bounded by $h(x^*, y^*)$ which is negative in this case. Note that the value of the first and third terms in $\Cref{eq:proofcase2}$ are fixed regardless of arbitrageur $i$'s report. Then it's sufficient to prove $\hat{x}_i \cdot v_i^* + \hat{y}_i$ is minimized when $(\hat{x}_i, \hat{y}_i) = (x_i^*, y_i^*)$ where $y_i^*/x_i^* = v_i^*$. Note that
    \begin{align*}
         \hat{x}_i \cdot v_i^* + \hat{y}_i &= (x_i^* + \hat{x}_i - x_i^*)\cdot v_i^* + (y_i^* + \hat{y}_i - y_i^*)  \\
         &= (x_i^* \cdot v_i^* + y_i^*) + (\hat{x}_i - x_i^*)\cdot v_i^* + (\hat{y}_i - y_i^*).
    \end{align*}
    By the assumption that $\left|\frac{\partial F/\partial x}{\partial F/\partial y}\right|$ is decreasing with respect to $x$ (see \cref{sec:model}), we have $\frac{\hat{y}_i - y_i^*}{x_i^* - \hat{x}_i} > v_i^*$ when $\hat{x}_i < x_i^*$, and $\frac{y_i^* - \hat{y}_i}{\hat{x}_i - x_i^*}< v_i^*$ when $\hat{x}_i > x_i^*$. Both cases imply that $(\hat{x}_i - x_i^*)\cdot v_i^* + (\hat{y}_i - y_i^*) > 0$ if $(\hat{x}_i, \hat{y}_i) \neq (x_i^*, y_i^*)$, and $\hat{x}_i \cdot v_i^* + \hat{y}_i$ reaches its minimum value when $(\hat{x}_i, \hat{y}_i) = (x_i^*, y_i^*)$.

    \textbf{Case 3:} Arbitrageur $i$ is the winner whether reporting truthfully or misreporting. In this case, we will show that after misreporting, the profit from the sub-bundle is no greater than before. If the sub-bundle corresponds to a transaction $\TX_j$, the profit for both scenarios is the same, which is $\Delta x_j \cdot v_i^* + \Delta y_j - \max_{k\neq i} V_k(\TX_j)$. If the sub-bundle corresponds to $s_0$, the profit after misreporting decreases because according to the analysis in Case 2, $h(\hat{x}_i, \hat{y}_i) < h(x^*, y^*)$ when $(\hat{x}_i, \hat{y}_i) \neq (x_i^*, y_i^*)$. 

    To sum up, arbitrageur $i$'s best strategy is to report truthfully by setting $q_i = v_i^*$, namely, 
    \begin{center}
        $u(\emptyset,v_i^*;S_{-i},\mathbf{q}_{-i})\geq u(\emptyset,q_i;S_{-i},\mathbf{q}_{-i})$ for all $q_i\in \mathbb{R}$.
    \end{center}
    This concludes the proof.
\end{proof}

\subsection{Proof of \Cref{thm:sybilproof}}\label{subsec:oursybilproof}
\begin{proof}[Proof of \Cref{thm:sybilproof}]
    Fix an arbitrary arbitrageur $i$ and the reports $\textbf{q}$ of all arbitrageurs. Whether a transaction $\TX_j \in M$ can be included in the bundle and thus executed only depends on $\max_{k\in [n]} \Delta x_j \cdot q_k + \Delta y_j$, where $(\Delta x_j, \Delta y_j)$ is determined by $\TX_j$ itself (see \Cref{eq:impact}). Suppose the maximal value is $\Delta \phi_w$ achieved by arbitrageur $w \in [n]$. If $\Delta \phi_w \geq 0$, $\TX_j$ will be deterministically executed at its limit state, causing a deterministic change in the user's account balance, and receive a refund $\max_{k \neq w} \max\{0, \Delta x_j \cdot q_k + \Delta y_j\}$. As a result, the number of tokens that the user of $\TX_j$ owns after executing through our mechanism is only determined by the transaction itself and $\textbf{q}$, which are fixed and independent of arbitrageur $i$'s strategy $S_i$. This concludes the proof.
\end{proof}

\subsection{Proof of \Cref{thm:NE}}\label{subsec:NEproof}
\begin{proof}[Proof of \Cref{thm:NE}]
    Our proof strategy below follows from two steps: We first give a comprehensive analysis of arbitrageur $i$'s profit with arbitrary strategy given everyone's belief $v_i^*$. The analysis will provide the intuition of our definition of the Sybil strategy. We then formally define the Sybil strategy $S_i(v_i^*,{b}_i^{\calX}, {b}_i^{\calY}, \calD)$ and show that $(v_i^*, S_i(v_i^*,{b}_i^{\calX}, {b}_i^{\calY}, \calD))$ is a best strategy of arbitrageur $i$ when everyone else follows this strategy and their belief $v_k^*$ is drawn from $\calD_k$. 
	
    \paragraph{First step.}	Fix an arbitrageur $i$, its true belief $v_i^*$, its budget $({b}_i^{\calX}, {b}_i^{\calY})$, and the strategies of all other arbitrageurs $\set{(v_k^*, S_k(v_k^*,{b}_k^{\calX}, {b}_k^{\calY}, \calD))}_{k \neq i}$. Here let's say $S_k(v_k^*,{b}_k^{\calX}, {b}_k^{\calY}, \calD)$ is some abstract Sybil strategy and it doesn't affect the analysis in the first step. We use $S_k$ as a shortening of $S_k(v_k^*,{b}_k^{\calX}, {b}_k^{\calY}, \calD)$ for every $k\in[n]$ for simplicity of notations.

	Given an arbitrary arbitrageur $i$'s strategy $(q_i, S_i')$, $i$'s utility $u(S_i', q_i; S_{-i}, \mathbf{v}^*_{-i})$ can be calculated by enumerating all sub-bundles in the output bundle of our mechanism. Our analysis and proof will be based on analyzing the profit of two different groups of the sub-bundles.

The proof of \Cref{theorem:truthful} implies that for every sub-bundle for which the middle transaction $e\in R\cup S_{-i}\cup\{s_0\}$, we have the profit from reporting $q_i$ is no more than the profit from reporting $v_i^*$. Thus, we will be mainly focusing on the arbitrageur $i$'s profit of the sub-bundle of every $e\in S_i'$ when $i$ reports $q_i$. 

Let's consider any $e\in S_i'$ such that $e=\TX=(\calX \to \calY, \delta_{\calX}^{in}, \delta_{\calY}^{out})$. The other case will be similar. Note that if $\TX$ is sandwiched by $i$ itself (which means $q_i\geq \max_{k\neq i}\{v_k^*\}$), then the profit of $i$ is 0; if $\TX$ is sandwiched by someone else, denoted by $w$, then arbitrageur $i$'s utility from this sub-bundle is 
    \begin{equation}\label{equation:H}
    H(q_i,\delta_{\calY}^{out})=\underbrace{\delta_{\calY}^{out} - \delta_{\calX}^{in} \cdot v_i^*}_{\text{Execution of $\TX$}} + \underbrace{\max\bigg(0,\delta_{\calX}^{in} \cdot q_i - \delta_{\calY}^{out},\max_{k \neq i,w}\{\delta_{\calX}^{in} \cdot v_k^* - \delta_{\calY}^{out}\}\bigg)}_{\text{Refund to $\TX$}}.
    \end{equation}
    
    To analyze the profit formula above, note that $H(\cdot,\cdot)$ is an increasing function of both parameters. However, there are also two upper bounds of these parameters based on the fact that this sub-bundle is won and sandwiched by $w\neq i$:
    \begin{itemize}
    	\item $w$ is the winner of all arbitrageur, which means $q_i\leq v_w^*$;
    	\item the profit of $w$ is no less than 0, which means $\delta_{\calY}^{out} \leq \delta_{\calX}^{in} \cdot v_w^*$.
    \end{itemize}
    We need the further property of $H(\cdot,\cdot)$, stated below:
    \begin{claim}\label{claim:H property}
    	For any $t$ such that $t\leq v_w^*$, we have $H(t, \delta_{\calX}^{in} \cdot t)=H(q_i, \delta_{\calX}^{in}\cdot t)$ for all $q_i\leq t$. 
    \end{claim}
    \begin{proof}
    Note that for the parameter regime that we are considering, we have $\delta_{\calX}^{in} \cdot q_i - \delta_{\calX}^{in} \cdot t \leq 0$. Thus $$H(q_i, \delta_{\calX}^{in} \cdot t)=H(t, \delta_{\calX}^{in} \cdot t) = \delta_{\calX}^{in} \cdot (t-v_i^*) + \max\big(0,\max_{k \neq i,w}\{\delta_{\calX}^{in} \cdot v_k^* - \delta_{\calY}^{out}\}\big).$$
    \end{proof}
    
    The claim above essentially says that, the arbitrageur $i$'s profit from \textit{its own Sybil transactions} is the same for any report $q_i \leq \delta_{\calY}^{out}/\delta_{\calX}^{in}$ regardless how small it is. Note that the arbitrageur $i$ also has the profit from users' transactions. Thus this provides a good intuition about what kind of Sybil strategy everyone is using could form a Nash equilibrium. We formalize it in the second step.

    \paragraph{Second step.}
    
We first specify the Sybil strategy $S_i(v_i^*,{b}_i^{\calX}, {b}_i^{\calY}, \calD)$ for every arbitrageur $i$. Fix an arbitrageur $i$, it includes two Sybil transactions: $\TX=(\calX \to \calY, \delta_{\calX}^{in} = {b}_i^{\calX}, \delta_{\calY}^{out}=t_{\calY}^*)$ and $\TX=(\calY \to \calX, \delta_{\calY}^{in} = {b}_i^{\calY}, \delta_{\calX}^{out}=t_{\calX}^*)$. To specify the $t_{\calY}^*$ above, we recall the profit function of $i$ given $q_i$ and $\delta_{\calY}^{out}$ as parameters:

\begin{equation*}
	\Gamma(q_i,\delta_{\calY}^{out},v_{-i}^*)_i=
	\left\{ \begin{array}{ll}
	H(q_i,\delta_{\calY}^{out}) & q_i\leq \max_{k \neq i}\{ v_k^* \} \text{ and } \delta_{\calY}^{out}\leq b_i^{\calX} \cdot \max_{k \neq i}\{ v_k^* \};\\
	0 & \text{otherwise}.
	\end{array}
	\right.
	\end{equation*}
Importantly, note that in the definition of $\Gamma()$ above, we should use ${b}_i^{\calX}$ as the parameter of $\delta_{\calX}^{in}$ when we refer the $H$ function of \Cref{equation:H}.

Now we are ready to define $t_{\calY}^*$ and $t_{\calX}^*$ can be defined similarly.
\begin{equation*}
        t_{\calY}^*\coloneqq \arg\max_t \mathbb{E}_{\bm{v}_{-i}^*\sim \calD_{-i}}\left[\Gamma(v_i^*,t,\bm{v}_{-i}^*)_i \right].
\end{equation*}
Namely, we choose $t_{\calY}^*$ that maximizes the expected profit of \textit{Sybil transactions} given that \textit{$i$ reports $v_i^*$ truthfully}. Note that by \Cref{theorem:truthful}, reporting truthfully can maximize arbitrageur $i$'s profit of transactions that are not $i$'s Sybil transactions. Thus, we only need to show that arbitrageur $i$'s profit of \textit{its Sybil transactions} is maximized when $i$ uses our specified strategy, comparing with arbitrary report $q_i$ and set of Sybil transactions $S_i'$. We consider the case where $S_i'$ only contains one Sybil transaction of $(\calX \to \calY, \delta_{\calX}^{in} = {b}_i^{\calX}, \delta_{\calY}^{out})$ (and only contains one Sybil transaction of $(\calY \to \calX, \delta_{\calY}^{in} = {b}_i^{\calY}, \delta_{\calX}^{out})$). This is without loss of generality because we can easily merge multiple Sybil transactions in the same direction. It is easy to see that $\delta_{\calY}^{out}$ should be at least ${b}_i^{\calX}\cdot v_i^*$.

It remains for us to show that $\mathbb{E}_{\bm{v}_{-i}^*\sim \calD_{-i}}\left[\Gamma(v_i^*,t_{\calY}^*,\bm{v}_{-i}^*)_i \right]\geq \mathbb{E}_{\bm{v}_{-i}^*\sim \calD_{-i}}\left[\Gamma(q_i,\delta_{\calY}^{out},\bm{v}_{-i}^*)_i \right]$. The proof will be purely based on the properties of the $H$ function. Recall that $H$ is increasing for both parameters and $\Gamma$ is non-zero only if $q_i\leq \max_{k \neq i}\{ v_k^* \} \text{ and } \delta_{\calY}^{out}\leq b_i^{\calX} \cdot \max_{k \neq i}\{ v_k^* \}$. Letting $\alpha=\max(q_i,\delta_{\calY}^{out}/b_i^{\calX})$, we have that $$\mathbb{E}_{\bm{v}_{-i}^*\sim \calD_{-i}}\left[\Gamma(\alpha,b_i^{\calX}\cdot \alpha,\bm{v}_{-i}^*)_i \right]\geq \mathbb{E}_{\bm{v}_{-i}^*\sim \calD_{-i}}\left[\Gamma(q_i,\delta_{\calY}^{out},\bm{v}_{-i}^*)_i \right].$$
Furthermore, by Claim~\ref{claim:H property}, we know that 
$$\mathbb{E}_{\bm{v}_{-i}^*\sim \calD_{-i}}\left[\Gamma(v_i^*,b_i^{\calX}\cdot \alpha,\bm{v}_{-i}^*)_i \right]\geq \mathbb{E}_{\bm{v}_{-i}^*\sim \calD_{-i}}\left[\Gamma(\alpha,b_i^{\calX}\cdot \alpha,\bm{v}_{-i}^*)_i \right].$$
Finally, by the definition of $t_{\calY}^*$, we conclude 
$$\mathbb{E}_{\bm{v}_{-i}^*\sim \calD_{-i}}\left[\Gamma(v_i^*,t_{\calY}^*,\bm{v}_{-i}^*)_i \right]\geq \mathbb{E}_{\bm{v}_{-i}^*\sim \calD_{-i}}\left[\Gamma(v_i^*,b_i^{\calX}\cdot \alpha,\bm{v}_{-i}^*)_i \right].$$

The optimal choice of $t_{\calX}^*$ and its analysis is analogous. This concludes the proof.
\end{proof}

\section{Implementation Considerations}
\label{app:implementation}

Below, we outline the architecture and workflow of a TEE-based implementation. We assume TEEs achieve confidentiality and integrity; in practice, one must carefully deal with side-channel attacks.

\subsection{Overview}
Figure~\ref{fig:workflow} illustrates the architecture of \AMMName.  As shown, there are four main entities: users, arbitrageurs, a TEE, and a smart contract.

\begin{figure}
    \centering
    \includegraphics[width=0.9\linewidth]{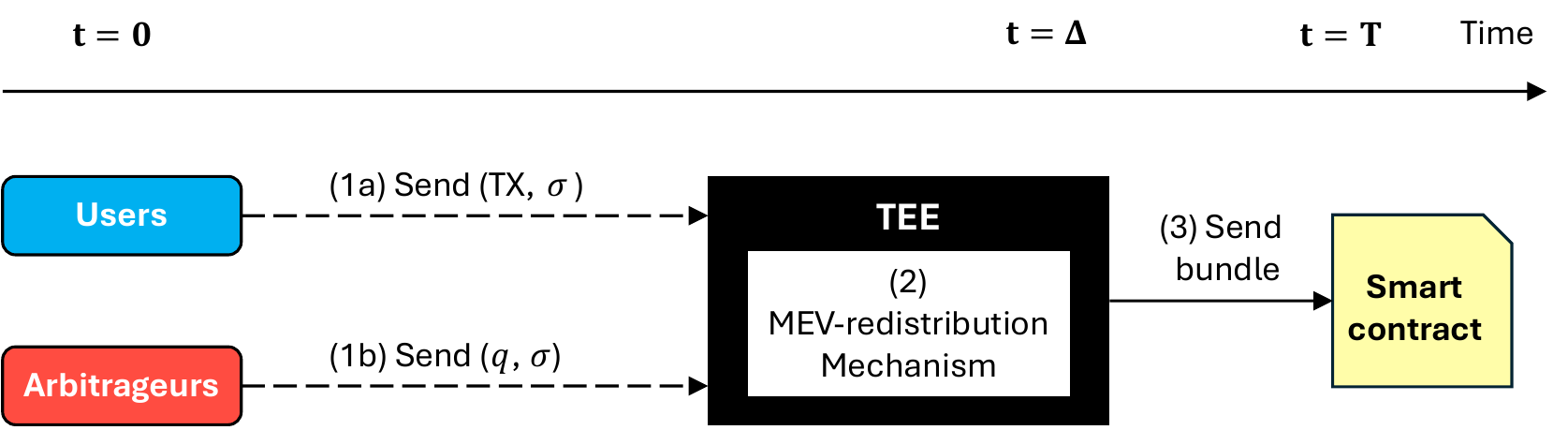}
    \caption{Overview of the \AMMName architecture and workflow. Users and arbitrageurs privately submit transactions and reports to the TEE, where the MEV-redistribution mechanism operates. Once the registration phase ends (as determined by the parameter $\Delta$), the TEE executes the MEV-redistribution process, constructs a bundle containing the swap, payment, and refund details, and then sends it publicly to trigger on-chain operations by invoking smart contracts.}
    \label{fig:workflow}
\end{figure}

\begin{itemize}
    \item \textbf{Users} are traders who use \AMMName for asset swaps. 
    They submit transactions to the TEE, specifying the swap direction, the maximum input amount, and the minimum output amount.
    \item \textbf{Arbitrageurs} are a combination of MEV searchers, market makers, and/or other participants seeking non-atomic arbitrage opportunities. Arbitrageurs must grant the \AMMName smart contract permission to manage a portion of their token holdings in advance. They provide a signed message to the TEE indicating their belief in the external price of the risk asset. 
    
    \item \textbf{TEE} processes submissions from users and arbitrageurs, running the MEV-redistribution mechanism to form bundles that include token swaps (i.e., outcome of the bundle generation rule) and transfers (for payments and refunds). The bundle is forwarded to the \AMMName smart contract for execution. %
    
    \item \textbf{\AMMName smart contract} executes the bundle from the TEE by interacting with the transfer and swap functions to move tokens among users, arbitrageurs, and the CFMM pool.

\end{itemize}

\subsection{Workflow}
We now elaborate on the workflow in Figure~\ref{fig:protocolcode}. \AMMName operates in slots, each with a duration $T$ consistent with the block time (e.g., $T=12s$ on Ethereum). Each slot consists of two stages: the registration stage and the execution stage. During the first registration stage, users and arbitrageurs submit their signed messages to the TEE, indicating their desire to participate. During the second execution stage, the TEE processes these messages, executes the MEV-redistribution mechanism, constructs the bundle accordingly, and invokes the smart contract to complete on-chain operations. The detailed steps of the protocol execution are as follows:

\begin{enumerate}
    \item[(1a)] A user $\puser_j$ prepares a transaction $\TX_j$, specifying the swap direction ($\calX \to \calY$ or $\calY \to \calX$), a maximum input amount $\delta^{in}$, and a minimum output amount $\delta^{out}$.
    The user signs the order with their private key $\mathsf{sk}_{\puser_j}$, producing the signature $\sigma_{\puser_j} \coloneqq \textsf{Sig}(\mathsf{sk}_{\puser_j}, \TX_j)$. The user then sends a message, $(\msgorder, (\TX_j, \sigma_{\puser_j}))$, to the TEE as described in Lines 1-4 of $\protoAMM$. 
    
    \item[(1b)] An arbitrageur $\parbtrgr_i$ provides a report $q_i$ on the external price of the risky asset $\calX$. 
    Like users, the arbitrageur signs the report with their private key $\mathsf{sk}_{\parbtrgr_i}$, generating a signature $\sigma_{\parbtrgr_i} \coloneqq \textsf{Sig}(\mathsf{sk}_{\parbtrgr_i}, q_i)$. The arbitrageur sends the signed message, $(\msgarbitrage, (q_i, \sigma_{\parbtrgr_i}))$, to the TEE, as specified in Lines 5-8.
    
    \item[(2)] Upon receiving the messages from users and arbitrageurs, the TEE maintains two lists: $M$ for pending orders and $\mathbf{q}$ for arbitrage reports. When the protocol-defined time $\Delta$ ($< T$) is reached, the registration stage ends. The TEE then executes the MEV-redistribution mechanism based on the pool's initial state $s_0$, pending transactions $M$, and the arbitrageurs' reports $\mathbf{q}$. It starts by applying Algorithm~\ref{alg:ourMechanism} (the bundle generation rule) to produce an MEV bundle. Next, the TEE implements the payment and refund rules. For each MEV opportunity from a pending transaction $\TX_i$ or the initial state $s_0$, the TEE constructs a transfer transaction from the winner $\mathcal{A}_{w_i}$ to the owner of $\TX_i$ or the pool (i.e., LPs), and adds it to the bundle. The whole process is indicated in Lines 9-30.
    
    \item[(3)] The TEE sends the final bundle to the \AMMName smart contract and calls the \textsf{settle()} function. This function parses the bundle and proceeds each transaction sequentially: for trades, it invokes the pool's \textsf{swap()} function; for token transfers, it calls the \textsf{transfer()} function. This completes the settlement phase, ensuring all token transfers and trades are executed on-chain.
\end{enumerate}

\begin{figure*}
\protocolsidebyside
{$\protoAMM(\Delta, pool)$}
{
\textcolor{mypurple}{\underline{\textbf{Users $\mathcal{P}_j$}}}: \\
    \t $\TX_j =$ (direction, $\delta^{in}$, $\delta^{out}$) \ \pccomment{swap} \\
    \t $\sigma_{\puser_j} \coloneqq \textsf{Sig}(\mathsf{sk}_{\puser_j}, \TX_j)$ \\
    \t send $(\msgorder, (\TX_j, \sigma_{\puser_j}))$ to \textcolor{mypurple}{\textbf{TEE}}  \\[1mm]
\textcolor{mypurple}{\underline{\textbf{Arbitrageurs $\mathcal{A}_i$}}}: \\
    \t $q_i$ = $\parbtrgr_i$'s report on the external price \\  
    \t $\sigma_{\parbtrgr_i} \coloneqq \textsf{Sig}(\mathsf{sk}_{\parbtrgr_i}, q_i)$ \\
    \t send $(\msgarbitrage, (q_i, \sigma_{\parbtrgr_i}))$ to \textcolor{mypurple}{\textbf{TEE}}  \hspace{1cm}  \\[1mm]
\textcolor{mypurple}{\underline{\textbf{TEE}}}: \\
$s_0 \gets$ $pool$'s initial state    \\
Initialize $M = [ \ ]$ \pccomment{pending transactions} \\
Initialize $\mathbf{q} = [ \ ]$  \ \pccomment{arbitrageurs in waiting} \\
\onrecv $(\msgorder, (\TX_j, \sigma_{\puser_j}))$: \\
    \t $M = M \cup \TX_j$ \\[1mm]
\onrecv $(\msgarbitrage, (q_i, \sigma_{\parbtrgr_i}))$: \\
    \t $\mathbf{q} = \mathbf{q} \cup q_i$ 
}
{
\textcolor{mypurple}{\underline{\textbf{TEE}}} (Cont'd): \\
\ontime ($\Delta$): \\
    \t \pccomment{execute the MEV-redistribution mechanism}  \\
    \t \textsf{bundle} $\gets$ Algorithm~\ref{alg:ourMechanism} $(s_0, M, \mathbf{q})$ \\
    \t \textbf{for} $i \in [|M| + 1]$  \ \pccomment{payment and refund rules}  \\
    \t\t $\mathcal{A}_{w_i} \gets winner$, amount $\gets \mathcal{A}_{w_i}$'s payment  \\
    \t\t to $\gets$ owner of $\TX_i$ or \textit{pool}'s address   \\
    \t\t $\TX_{w_i}$ = (token, amount, to)  \ \pccomment{transfer} \\
    \t\t \textsf{bundle}.append$(\TX_{w_i})$    \\
    \t call \textcolor{magenta}{\AMMName.settle} (\textsf{bundle}) \\[1mm]
\textcolor{mypurple}{\underline{\textbf{\AMMName Smart Contract}}}:  \\
\MCMsettle(\textsf{info}):   \\
    \t \textbf{for} $\TX \in \textsf{bundle}$    \\
    \t\t \textbf{if} $\TX.type$ = \textit{swap} \textbf{then}  \\
    \t\t\t call \textcolor{magenta}{pool.swap}(direction, $\delta^{in}$, $\delta^{out}$)    \\
    \t\t \textbf{else}  \\
    \t\t\t call \textcolor{magenta}{token.transfer}(token, amount, to)
}
\caption{\AMMName Protocol.}
\label{fig:protocolcode}
\end{figure*}

\end{document}